\documentclass[journal,twoside,web]{ieeecolor}
\usepackage{generic}
\usepackage{cite}
\usepackage{textcomp}

\usepackage{amssymb,latexsym,amsfonts,amsmath,bbm}
\usepackage[cal=cm,bb=pazo]{mathalfa}
\usepackage{mathrsfs}
\usepackage{graphicx}
\usepackage{paralist}
\usepackage{dsfont}
\usepackage{eso-pic}
\usepackage{epsfig}
\usepackage{mathtools}
\usepackage{epstopdf}
\usepackage{lipsum}
\usepackage{cite}
\usepackage{caption}
\usepackage{subcaption}

\usepackage{booktabs}
\usepackage{multirow}

\usepackage[nocomma]{optidef}

\usepackage{tikz}
\usetikzlibrary{calc,shapes,arrows}

\usepackage{algorithm}
\usepackage{algorithmic}

\newcounter{theorem}
\newcounter{definition}
\newcounter{lemma}
\newcounter{claim}
\newcounter{problem}
\newcounter{proposition}
\newcounter{corollary}
\newcounter{construction}
\newcounter{example}
\newcounter{xca}
\newcounter{comments}
\newcounter{remark}
\newcounter{assumption}

\newtheorem{theorem}[theorem]{Theorem}
\newtheorem{lemma}[lemma]{Lemma}

\newtheorem{problem}[problem]{Problem}

\newtheorem{definition}[definition]{Definition}

\newtheorem{remark}[remark]{Remark}
\newtheorem{assumption}[assumption]{Assumption}

\numberwithin{equation}{section}

\DeclareFontFamily{U}{stix2bb}{}
\DeclareFontShape{U}{stix2bb}{m}{n} {<-> stix2-mathbb}{}

\usepackage[many]{tcolorbox}
\usetikzlibrary{calc}
\tcbuselibrary{skins}

\newtcolorbox{resp}[1][]{%
	enhanced jigsaw,%
	colback=gray!5!white,%
	colframe=gray!80!black,%
	size=small,%
	boxrule=1pt,%
	halign title=flush center,%
	coltitle=black,%
	breakable,%
	drop shadow=black!50!white,%
	attach boxed title to top left={xshift=1cm,yshift=-\tcboxedtitleheight/2,yshifttext=-\tcboxedtitleheight/2},%
	minipage boxed title=3cm,%
	boxed title style={%
		colback=white,%
		size=fbox,%
		boxrule=1pt,%
		boxsep=2pt,%
		underlay={%
			\coordinate (dotA) at ($(interior.west) + (-0.5pt,0)$);
			\coordinate (dotB) at ($(interior.east) + (0.5pt,0)$);
			\begin{scope}[gray!80!black]
				\fill (dotA) circle (2pt);
				\fill (dotB) circle (2pt);
			\end{scope}
		}%
	},%
	#1%
}

\usepackage{xspace}

\definecolor{blush}{rgb}{0.87, 0.36, 0.51}

\definecolor{blue(ryb)}{rgb}{0.01, 0.28, 1.0}
\definecolor{fashionfuchsia}{rgb}{0.96, 0.0, 0.63}
\definecolor{byzantine}{rgb}{0.74, 0.2, 0.64}
\definecolor{electricviolet}{rgb}{0.56, 0.0, 1.0}
\makeatletter
\let\NAT@parse\undefined
\makeatother
\usepackage[colorlinks=true, citecolor=electricviolet, linkcolor=blue(ryb), final]{hyperref}

\renewcommand{\theequation}{\arabic{equation}}
\counterwithout*{equation}{section}

\makeatletter
\def\@opargbegintheorem#1#2#3{\textit{#1\ #2} \textit{(#3):}}
\makeatother

\definecolor{START}{rgb}{1.00, 0.76, 0.07}
\definecolor{TARGET}{rgb}{0.00, 0.50, 0.50}
\definecolor{OBSTACLES}{rgb}{1.00, 0.44, 0.38}

\definecolor{START1}{rgb}{0, 1, 0}
\definecolor{TARGET1}{rgb}{0.3010 0.7450 0.9330}
\definecolor{OBSTACLES1}{rgb}{1, 0, 0}

\usepackage{microtype}

\begin{document}
	
\title{From Dissipativity Property to Data-Driven GAS Certificate of Degree-One Homogeneous Networks with Unknown Topology}
\author{Abolfazl Lavaei, \IEEEmembership{Senior Member,~IEEE}, and David Angeli, \IEEEmembership{Fellow,~IEEE}	
	\thanks{A. Lavaei is with the School of Computing, Newcastle University, United Kingdom. D. Angeli is with the Department of Electrical and Electronic Engineering, Imperial College London, United Kingdom, and with the Department of Information Engineering, University of Florence, Italy. Email: \texttt{abolfazl.lavaei@newcastle.ac.uk}, \texttt{d.angeli@imperial.ac.uk}.}
}

\maketitle
\begin{abstract}
	In this work, we propose a data-driven \emph{divide and conquer} strategy for the stability analysis of interconnected homogeneous nonlinear networks of degree one with \emph{unknown models} and a \emph{fully unknown topology}. The proposed scheme leverages joint dissipativity-type properties of subsystems described by
	\emph{storage functions}, while providing a stability certificate over unknown interconnected networks. In our data-driven framework, we begin by formulating the required conditions for constructing storage functions as a robust convex program (RCP). Given that unknown models of subsystems are integrated into one of the constraints of the RCP, we collect data from trajectories of each unknown subsystem and provide a scenario convex program (SCP) that aligns with the original RCP. We solve the SCP as a linear program and construct a storage function for each subsystem with unknown dynamics. Under some newly developed \emph{data-driven compositionality conditions}, we then construct a \emph{Lyapunov function} for the \emph{fully unknown interconnected network} utilizing storage functions derived from data of individual subsystems. We show that our data-driven {divide and conquer strategy} provides correctness guarantees (as opposed to probabilistic confidence) while significantly mitigating the sample complexity problem existing in data-driven approaches. To illustrate the effectiveness of our proposed results, we apply our approaches to three different case studies involving interconnected homogeneous (nonlinear) networks with \emph{unknown models}. We collect data from trajectories of unknown subsystems and verify the \emph{global asymptotic stability (GAS)} of the interconnected system with a guaranteed confidence.
\end{abstract}


\section{Introduction}

In recent years, data-driven analysis of large-scale networks has become increasingly prevalent, driven by their extensive applications in real-world safety-critical systems. Notably, many complex and heterogeneous systems lack closed-form mathematical models that are either unavailable or too intricate to be practically useful. As a result, the model-based techniques proposed in relevant literature cannot effectively analyze such complex networks. While there are \emph{indirect} data-driven approaches based on identification techniques to analyze unknown systems by learning their approximate models (see \emph{e.g.,}~\cite{Hou2013model,wang2018safe,cheng2019end,sadraddini2018formal,lindemann2020learning}), obtaining accurate models for complex systems remains highly challenging, time-consuming, and expensive. Furthermore, in the realm of stability analysis, even if a model can be identified using system identification techniques, there remains the necessity to find a Lyapunov function for the acquired model to demonstrate stability properties. Consequently, computational complexity arises in both stages: identifying the model and constructing a Lyapunov function. 
Therefore, it is crucial to develop \emph{direct} data-driven approaches to bypass the system identification phase and directly utilize system measurements for conducting the underlying analysis~\cite{dorfler2022bridging} (cf. Fig.~\ref{Fig100}).

\begin{figure} 
	\begin{center}
		\includegraphics[width=0.87\linewidth]{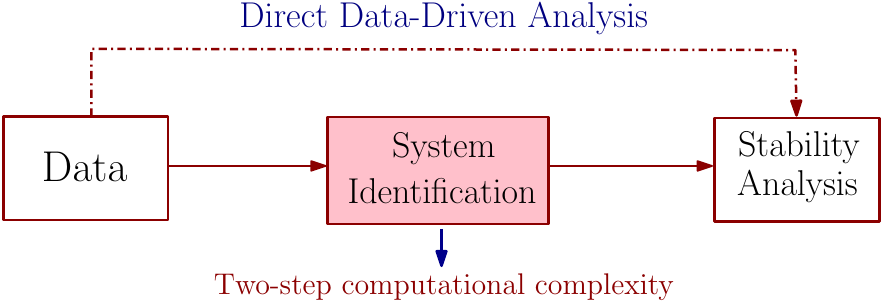} 
		\caption{Indirect (\emph{i.e.,} system identification) vs. direct data-driven techniques.} 
		\label{Fig100}
	\end{center}
\end{figure}

Over the past two decades, several advancements have emerged in the realm of data-driven optimization techniques. Notably, the scenario approach was initially introduced in~\cite{calafiore2006scenario} for addressing robust control analysis and synthesis problems through semi-infinite convex programming. By introducing a probabilistic relaxation of the worst-case robust paradigm, the robust control problem outlined in~\cite{calafiore2006scenario} can be tackled via random sampling of constraints. Building upon the foundation laid in~\cite{calafiore2006scenario}, the work~\cite{calafiore2010random} presents a framework for random convex programs, offering a key advantage over~\cite{calafiore2006scenario} by establishing an explicit bound on the upper tail probability of violation. Furthermore, this framework is extended to encompass random convex programs with posteriori violated constraints, aimed at enhancing the optimal objective value while effectively managing the violation probability.

An initial probabilistic connection between the optimal value of scenario convex programs and that of robust convex and chance-constrained programs is introduced in~\cite{esfahani2014performance,margellos2014road}, where the presented findings in~\cite{esfahani2014performance} are also extended to a specific class of non-convex problems involving binary decision variables. Furthermore,~\cite{mohajerin2018infinite} presents a constructive framework for establishing a relationship between an infinite-dimensional linear programming and a manageable finite convex program. The proposed results in this work are founded on randomized optimization and first-order methods, offering both a priori and posteriori performance guarantees.

In recent years, there have been notable advancements in the formal stability analysis of unknown dynamical systems through data-driven approaches. In this respect, \cite{kenanian2019data} presents a data-driven stability analysis of black-box switched linear systems by constructing a \emph{common} Lyapunov function, offering a stability guarantee based on both the number of observations and a required level of confidence. Building upon this work, \cite{rubbens2021data} extends the approach to propose a data-driven stability analysis of switched linear systems using \emph{sum of squares} Lyapunov functions. Data-driven stability verification of homogeneous  systems with a probabilistic level of confidence is proposed in \cite{Lavaei_Data_CDC2022}. Additionally, \cite{boffi2020learning} introduces a data-driven approach for the stability analysis of continuous-time unknown systems. In~\cite{de2019persistency}, a parameterization scheme for linear-feedback systems utilizing data-dependent linear matrix inequalities is introduced, extending the stabilization problem to include output-feedback control design. Meanwhile,~\cite{coulson2019regularized} presents a data-enabled predictive control algorithm for unknown stochastic linear systems. This approach utilizes noise-corrupted input/output data to predict future trajectories and compute optimal control policies.  A data-driven adaptive second-order sliding mode control approach is presented in~\cite{samari2025data}, which guarantees semi-global asymptotic stability of the origin in the presence of disturbances. In~\cite{steentjes2021h}, an $H_{\infty}$ performance analysis and distributed controller synthesis for interconnected linear systems were proposed using noisy input-state data. Data-driven learning of incremental ISS controllers for unknown nonlinear polynomial dynamics has recently been proposed in~\cite{zaker2024certified}.

In contrast to existing literature, we introduce a data-driven approach for the stability analysis of \emph{large-scale networks} with  \emph{unknown models and interconnected topology}, employing a \emph{divide and conquer strategy}. While previous works such as~\cite{kenanian2019data,rubbens2021data,Lavaei_Data_CDC2022,boffi2020learning} focus solely on monolithic systems, they encounter significant challenges, known as \emph{sample complexity}, particularly in handling high-dimensional systems. This complexity manifests in the \emph{exponential} growth of the required data volume to provide stability certificates as the state space dimension increases. In contrast, utilizing our divide and conquer strategy, we break down the sample complexity to the level of subsystems, where the number of data grows \emph{linearly} with the number of subsystems.  Furthermore, while the stability guarantees in~\cite{kenanian2019data,rubbens2021data,Lavaei_Data_CDC2022,boffi2020learning} are based on probabilistic confidence levels, our stability results offer \emph{a guaranteed confidence of 1}.

Our primary  contribution is to develop a \emph{data-driven compositional approach} for analyzing the stability of interconnected homogeneous nonlinear networks, where each subsystem is homogeneous {of degree one} with unknown models. The proposed framework capitalizes on the joint dissipativity-type properties of subsystems, characterized by \emph{homogeneous storage functions} of degree 2. Within our data-driven context, we collect data from trajectory of each unknown subsystem and devise a scenario convex program (SCP) to enforce the required conditions of storage functions. By solving the SCP, we construct a data-driven storage function for each unknown subsystem. Subsequently, we construct a \emph{Lyapunov function} for the unknown interconnected network with a \emph{confidence level of 1} based on storage functions of individual subsystems, under a newly developed compositionality condition derived from data. We show that our innovative data-driven compositional approach eliminates the necessity to verify the traditional compositional condition based on dissipativity theory that requires the interconnected topology (cf.~\cite[Proposition 2.1]{2016Murat}). 

In the scenario where the interconnection topology is known, we provide a sub-result where we initially establish a storage function for each unknown subsystem based on data. This is achieved by solving an SCP for each unknown subsystem, ensuring a guaranteed confidence level of 1. Subsequently, we examine the traditional dissipativity compositionality condition to determine its satisfaction. This condition is represented as a linear matrix inequality (LMI) that interrelates individual storage functions. To avoid potential posteriori checks for this condition, we employ an alternating direction method of multipliers (ADMM) algorithm~\cite{boyd2011distributed,anand2024compositional}, enabling efficient satisfaction of the compositionality condition in a \emph{distributed manner} while searching for storage functions of individual subsystems.

A data-driven stability analysis of homogeneous networks was previously presented in~\cite{lavaei2023data}. Our paper offers several distinct advantages compared to~\cite{lavaei2023data}. First and foremost, while~\cite{lavaei2023data} provides stability guarantees based on ISS properties of subsystems, our focus here is on deriving stability certificates for interconnected networks through the \emph{dissipativity properties} of individual subsystems. Dissipativity is a broader framework that can simplify the conditions required for stability analysis. Specifically, our approach eliminates the need for the first radially unbounded condition in the ISS Lyapunov function (see~\cite[condition (6a)]{lavaei2023data}), allowing us to achieve stability with fewer constraints. Additionally, our proposed method involves a convex optimization program. In contrast, the optimization approach used in~\cite{lavaei2023data} involves bilinearity, which can be more complex and less tractable. This shift to a convex optimization framework in our work potentially offers computational advantages and makes our method more accessible for practical implementation.

Another key difference is in the treatment of the interconnection topology. While~\cite{lavaei2023data} assumes that the interconnection topology is known, our work considers the more challenging scenario where the topology is fully unknown—a situation common in many real-world applications. Given this unknown topology, one cannot directly verify the compositionality condition in~\cite[circularity condition (9)]{lavaei2023data}. Instead, we introduce a \emph{more relaxed} compositionality condition that can be validated during the optimization process using data, thus circumventing the need for prior knowledge of the network topology. Moreover, our approach reduces computational overhead by integrating the verification of the compositionality condition directly into the optimization problem. In contrast, the method in~\cite{lavaei2023data} requires posterior checks of the compositionality condition after solving the optimization problem and constructing ISS gains, which can involve significant computational costs (see~\cite[circularity condition (9)]{lavaei2023data}). Finally, we employ an algorithm to estimate the Lipschitz constant from data, whereas~\cite{lavaei2023data} relies on analytical computation based on specific system knowledge (see~\cite[Lemma 5.4]{lavaei2023data}). Our data-driven approach simplifies the process and avoids the need for system knowledge.

While the studies~\cite{romer2019one,van2022data,maupong2017lyapunov} investigate dissipativity properties based on data, they all focus on \emph{monolithic systems}, which face challenges with sample complexity when applied to large-scale dynamics. In contrast, our approach considers a network of thousands of subsystems, and aims to provide stability guarantees for the overall network by leveraging the dissipativity properties of individual subsystems. In particular, we assume that each subsystem is dissipative with respect to a supply rate (cf. condition~\eqref{Eq:8_2}), meaning that the energy consumed by a subsystem exceeds the energy it generates. We design these supply rates for each subsystem and use them to extend stability guarantees across the interconnected network by verifying the compositionality condition in~\eqref{Con2}, constructed from the supply rates derived from data. Recently, ISS, incremental ISS, and robust analysis of interconnected networks have been investigated in~\cite{zaker2025data}, \cite{zaker2025data1}, and \cite{samari2025data1}, respectively. 

The remainder of the paper is organized as follows. Section~\ref{Problem_Description} introduces mathematical notations, and formally defines discrete-time nonlinear networks. In Section~\ref{Data: Dissipativity}, we present a stability certificate for interconnected networks based on dissipativity-type properties of subsystems described by a notion of storage functions. Section~\ref{Data-Driven} outlines our data-driven scheme to construct storage functions for unknown subsystems from data.
In Section~\ref{Guarantee_RCP}, we construct a Lyapunov function for the interconnected system based on storage functions of individual subsystems, derived from data, ensuring a guaranteed confidence level of 1. Section~\ref{ADMM} discusses the use of an ADMM algorithm to efficiently satisfy the compositionality condition in a distributed manner in the case of known interconnection topologies. In Section~\ref{Case_Study}, we demonstrate our data-driven approaches using a room temperature network with unknown dynamics. Finally, Section~\ref{Discussion} presents concluding discussions.

\section{Problem Description}\label{Problem_Description}

\subsection{Notation}
We employ the following notation throughout the paper. Sets of real, positive, and non-negative real numbers are denoted by $\mathbb{R},\mathbb{R}^+$, and $\mathbb{R}^+_0$, respectively. We denote sets of non-negative and positive integers by $\mathbb{N} := \{0,1,2,\ldots\}$ and $\mathbb{N}^+=\{1,2,\ldots\}$, respectively. Given $N$ vectors $x_i \in \mathbb{R}^{n_i}$, $x=[x_1;\ldots;x_N]$ denotes the corresponding column vector of dimension $\sum_i n_i$.
We represent the Euclidean norm of a vector $x\in\mathbb{R}^{n}$ by $\Vert x\Vert$, while the absolute value of $a\in\mathbb{R}$ is denoted by $\vert a\vert$.  The Frobenius norm of a matrix $P\in \mathbb R^{n\times m}$ is denoted by $\Vert P \Vert_F$. Given sets $X_i, i\in\{1,\ldots,N\}$,
their Cartesian product is denoted by $\prod_{i=1}^{N}X_i$. We define an $(n \times n)$ identity matrix by  $\mathds{I}_n$, while an $(n \times m)$ zero matrix is denoted by $\mathbf{0}_{n\times m}$. 
A function $\eta\!:\mathbb\mathbb \mathbb R_{0}^+\rightarrow\mathbb\mathbb \mathbb R_{0}^+,$ is said to be a class $\mathcal{K}$ function if it is continuous, strictly increasing, and $\eta(0)=0$. A function $\eta: \mathbb{R}_{0}^+ \times \mathbb{R}_{0}^+ \rightarrow \mathbb{R}_{0}^+$ is said to belong to class $\mathcal {KL}$ if, for each fixed $s$, the map
$\eta(r,s)$ belongs to class $\mathcal K$ with respect to $r$ and, for each fixed $r>0$, the map $\eta(r,s)$ is decreasing with respect to $s$, and
$\eta(r,s)\rightarrow 0$ as $s \rightarrow \infty$.

\subsection{Discrete-Time Nonlinear Systems}

In this work, we consider discrete-time nonlinear systems (dt-NS), as formalized in the following definition.

\begin{definition}
	A discrete-time nonlinear system (dt-NS) is characterized by the tuple
	\begin{equation}\label{Eq:11}
		\Sigma = (\mathbb R^{n},f),
	\end{equation}
	where:
	\begin{itemize}
		\item $\mathbb R^{n}$ is the state space of the system;
		\item $f\!:\mathbb R^{n}\rightarrow \mathbb R^{n}$ is a continuous function governing the system's state evolution, assumed to be {homogeneous of degree one, \emph{i.e.,} for any $\lambda > 0$ and $x\in \mathbb R^{n}$, $f(\lambda x) = \lambda f(x)$~\cite[Definition 1.1]{polyakov2020generalized}}. We presume the map $f$ remains unknown to us.
	\end{itemize}
	
	Given an initial state $x(0)\in \mathbb{R}^{n}$, the evolution of the state of $\Sigma$ can be characterized as
	\begin{align}\label{Eq:21}
		\Sigma\!:x(k+1)=f(x(k)),\quad \quad k\in\mathbb N.
	\end{align}
\end{definition}

In this work, our aim is to verify the \emph{global asymptotic stability (GAS)} of $\Sigma$ concerning its equilibrium point, as defined below.

\begin{definition}\label{GAS}
	A dt-NS $\Sigma$ is deemed \emph{globally asymptotically stable (GAS)} with respect to the equilibrium point $x = 0$ if $$ \Vert x(k) \Vert \leq \eta(\Vert x(0)\Vert ,k),$$
	for all $x(0)\in \mathbb R^n$ and $\eta \in \mathcal {KL}$, indicating the convergence of all solutions of $\Sigma$ to the origin as $k\rightarrow \infty$.
\end{definition}

In the following section, we analyze stability of the interconnected dt-NS by studying dissipativity properties of individual subsystems.

\section{Stability Certificate via Dissipativity Approach}\label{Data: Dissipativity}

Here, we investigate the stability of interconnected networks by considering dissipativity-type properties of subsystems described by storage functions~\cite{2016Murat}. To achieve this, we represent discrete-time nonlinear \emph{subsystems} as
\begin{equation}\label{Eq:1}
	\Sigma_i = (\mathbb R^{n_i},\mathbb R^{p_i},f_i),
\end{equation}
where $\mathbb R^{n_i}$ is the state set of the subsystem, $\mathbb R^{p_i}$ is the \emph{internal input} set of the subsystem (which is used for the sake of interconnection), and $f_i:\mathbb R^{n_i}\times \mathbb R^{p_i}\rightarrow \mathbb R^{n_i}$ is a continuous function which is homogeneous of degree one, \emph{i.e.,} for any $\lambda > 0$ and $x_i\in \mathbb R^{n_i},w_i\in \mathbb R^{p_i}$, $f_i(\lambda x_i, \lambda w_i) = \lambda f_i(x_i,w_i)$. Given an initial state $x_i(0)\in \mathbb R^{n_i}$ and input sequence $w_i(\cdot)\!: \mathbb N \rightarrow \mathbb R^{p_i}$, the evolution of the state of $\Sigma_i$ is characterized as
\begin{align}\label{Eq:2}
	\Sigma_i\!:{x}_i(k+1)=f_i(x_i(k),w_i(k)),\quad \quad k\in\mathbb N.
\end{align}

Given $\Sigma_i=(\mathbb R^{n_i},\mathbb R^{p_i},f_i), i\in\{1,\dots,M\}$, with a matrix $\mathcal M$ being the coupling between subsystems, the interconnection of $\Sigma_i$ is
$\Sigma=(\mathbb R^{n},f)$ as in~\eqref{Eq:11}, denoted by
$\mathcal{I}(\Sigma_1,\ldots,\Sigma_M)$, where $n = \sum_{i = 1}^M n_i$ and $f(x):=[f_1(x_1,w_1);\dots;f_M(x_M,w_M)]$, subjected to the
following interconnection constraint:
\begin{align}\label{Eq:4}
	\Big[w_{1};\cdots;w_{M}\Big]=\mathcal M\Big[x_1;\cdots;x_M\Big]\!.
\end{align}
As all subsystems $\Sigma_i, i\in\{1,\dots,M\}$, are homogeneous of degree one, the interconnected dt-NS $\Sigma$ maintains homogeneity with the same degree. The schematic representation of an interconnected dt-NS $\Sigma$ is illustrated in Fig.~\ref{Fig1}.

\begin{figure} 
	\begin{center}
		\includegraphics[width=0.7\linewidth]{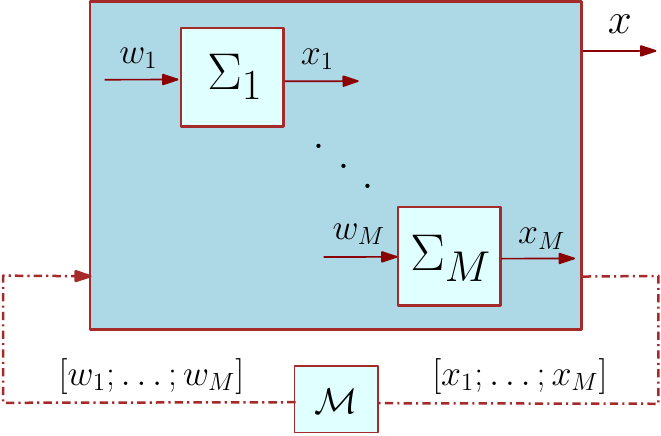} 
		\caption{Interconnected dt-NS $\mathcal{I}(\Sigma_1,\ldots,\Sigma_M)$.} 
		\label{Fig1}
	\end{center}
\end{figure}

To establish the stability certificate of interconnected dt-NS based on dissipativity reasoning, we initially define a concept of storage functions for each subsystem.

\begin{definition}\label{Def:4}
	Given a subsystem $\Sigma_i=(\mathbb R^{n_i},\mathbb R^{p_i},f_i)$, a positive-definite function $\mathcal{S}_i:\mathbb R^{n_i}\rightarrow\mathbb{R}_{0}^{+}$, \emph{i.e.,} $\mathcal{S}_i(0) = 0$ and $\mathcal{S}_i(x_i) > 0$ for any $x_i\in X_i \backslash 0$,  is called a storage function (SF) for $\Sigma_i$ if there exists a symmetric matrix $\mathcal X_i$ with partitions $\mathcal X_i^{jj'}$, $j,j'\in\{1,2\}$, such that:
	\begin{align}\notag
		&\forall x_i \in \mathbb R^{n_i}, \forall w_i \in \mathbb R^{p_i}\!:\\\label{Eq:8_2}
		&\mathcal S_i(f_i(x_i,w_i)) - \mathcal S_i(x_i) \leq \begin{bmatrix}
			w_i\\
			x_i
		\end{bmatrix}^\top\!
		\underbrace{\begin{bmatrix}
				\mathcal X_i^{11}&\mathcal X_i^{12}\\
				\mathcal X_i^{21}&\mathcal X_i^{22}
		\end{bmatrix}}_{\mathcal X_i}\!\begin{bmatrix}
			w_i\\
			x_i
		\end{bmatrix}\!\!.
	\end{align}
	The expression on the right-hand side of~\eqref{Eq:8_2} is referred to as the \emph{supply rate} of the system. We assume that storage functions $\mathcal{S}_i$ are homogeneous with degree $2$, \emph{i.e.,} for any $\lambda > 0$ and $x_i\in \mathbb R^{n_i}$, $\mathcal S_i(\lambda x_i) = \lambda^2 \mathcal S_i(x_i)$.
\end{definition}

The following theorem, adapted from~\cite[Proposition 2.1]{2016Murat}, offers the stability guarantee for the interconnected dt-NS based on dissipativity reasoning.

\begin{theorem}\label{Thm:3}
	Consider an interconnected dt-NS $\Sigma = \mathcal{I}(\Sigma_1,\ldots,\Sigma_M)$ composed of $M\in\mathbb{N}^+$ individual subsystems~$\Sigma_i$ and a coupling matrix $\mathcal M$ among them. Suppose that each $\Sigma_i$ admits an SF $\mathcal S_i$, as presented in Definition~\ref{Def:4}. If
	\begin{align}\label{EQQ:2}
		&\quad\begin{bmatrix}
			\mathcal M\\\mathds{I}
		\end{bmatrix}^\top \!\!\mathcal X_{cmp}\begin{bmatrix}
			\mathcal M\\\mathds{I}
		\end{bmatrix}\preceq 0, \\\notag
		\text{with} \quad 
		\mathcal X_{cmp}:=&\begin{bmatrix}
			\mathcal X_1^{11}&&&\mathcal X_1^{12}&&\\
			&\ddots&&&\ddots&\\
			&&\mathcal X_M^{11}&&&\mathcal X_M^{12}\\
			\mathcal X_1^{21}&&&\mathcal X_1^{22}&&\\
			&\ddots&&&\ddots&\\
			&&\mathcal X_M^{21}&&&\mathcal X_M^{22}
		\end{bmatrix}\!\!,
	\end{align}
	then $\mathcal V: \mathbb R^n\to \mathbb R^+_0$ in the form of \begin{align}\label{Lyapunov}
		\mathcal V(x) := \sum_{i=1}^{M}\mathcal S_i(x_i)
	\end{align}
	is a Lyapunov function for the interconnected system $\Sigma = \mathcal{I}(\Sigma_1,\ldots,\Sigma_M)$, satisfying the following two conditions:
	\begin{subequations}
		\begin{align}\label{alpha0}
			&\forall x \in \mathbb R^{n} \backslash \{0 \}\!:\quad 
			\mathcal V(x)>0,\\\label{alpha1}
			&\forall x \in \mathbb R^{n}\backslash \{0 \}\!:\quad 
			\mathcal V(f(x)) < \mathcal V(x).
		\end{align}
	\end{subequations}
	As a consequence, the interconnected system $\Sigma = \mathcal{I}(\Sigma_1,\ldots,\Sigma_M)$ is GAS with respect to $x = 0$, according to Definition~\ref{GAS}.
\end{theorem}

\begin{remark}
	Note that if the interconnection topology $\mathcal{M}$ in the model-based analysis is unknown, the only way to verify the compositionality condition in~\eqref{EQQ:2} is to ensure that the composed supply rate $\mathcal{X}_{\text{cmp}}$ is negative semi-definite, which is a highly conservative requirement. However, as demonstrated later in Section~\ref{Data-Driven}, our data-driven approach can handle networks with an unknown topology $\mathcal{M}$ without necessitating $\mathcal{X}_{\text{cmp}}$ to be negative semi-definite (cf. condition~\eqref{Con2}).
\end{remark}

In this work, our focus lies in constructing storage functions $\mathcal{S}_i$ to locally satisfy condition~\eqref{Eq:8_2} over the \emph{unit sphere} $\Vert(x_i,w_i)\Vert = 1, i\in\{1,\dots,M\}$. We then transfer the results for the interconnected network over the $\Vert x\Vert = 1$. Thanks to the \emph{homogeneity property} of the maps $f_i$ and storage functions $\mathcal{S}_i$, which ensure the homogeneity of the interconnected network, thereby the stability certificate can be globally guaranteed in $\mathbb{R}^n$.

\begin{remark}\label{Rem:10}
	To demonstrate that the Lyapunov function $\mathcal{V}$ for the interconnected network can be constructed using the storage functions $\mathcal{S}_i$ of individual subsystems, as shown in~\eqref{Lyapunov}, \emph{over the unit sphere}, it is required that $\Vert x\Vert = 1 \Leftrightarrow \Vert (x_i, w_i)\Vert = 1$. To ensure this, we assume that the state measurement of all subsystems in the interconnection topology is available, which is not a restrictive assumption. Now by considering  $w_i=[{x_1;\ldots;x_{i-1};x_{i+1};\ldots;x_{M}}]$, one can readily ensure that $\Vert (x_i, w_i)\Vert =1 \Leftrightarrow \Vert x\Vert  =1$. It is important to note that while $w_i$ is formulated based on the states of all other subsystems except itself, \emph{this does not necessitate full interconnection}. In fact, $w_i$ must incorporate the sates of other subsystems to maintain a well-posed setting over the unit sphere. However, the specific interconnection topology is determined by the values of the interconnection matrix $\mathcal M$ in~\eqref{Eq:4}, where the nonzero entries shape the nature of the interconnection.
\end{remark}

\begin{remark}
	It is worth noting that while the assumption of homogeneity is well-established in control theory and allows for the extension of the stability certificate globally to $\mathbb{R}^n$, it might  limit the class of systems considered. Relaxing this condition would enable the consideration of a broader class of systems; however, it would come at the cost of providing only \emph{local}, rather than global, stability certificates. We would like to emphasize that homogeneity of degree one is not restricted only to linear systems, even though all linear systems are homogeneous of degree one. We refer readers to the third case study in Subsection~\ref{third}, where we demonstrate that our framework can successfully handle a \emph{nonlinear} homogeneous system of degree one.
\end{remark}

Given the unknown dynamics of each subsystem $f_i$ in our setting, it is impractical to search for storage functions $\mathcal{S}_i$ that fulfill condition~\eqref{Eq:8_2}. Moreover, since the interconnection topology $\mathcal{M}$ is fully unknown, directly verifying the compositionality condition in~\eqref{EQQ:2} is not feasible. To address these challenges, we propose a \emph{direct data-driven} approach that introduces a novel compositional condition based on data to construct a Lyapunov function for unknown interconnected dt-NS in the form of~\eqref{Lyapunov} without checking the traditional dissipativity condition~\eqref{EQQ:2}. This Lyapunov function satisfies conditions~\eqref{alpha0}-\eqref{alpha1} and allows us to formally verify the GAS property of unknown interconnected dt-NS.

We will now articulate the central problem that this work aims to solve.\vspace{0.2cm} 

\begin{resp}
	\begin{problem}\label{problem1}
		Consider an unknown interconnected dt-NS $\Sigma = \mathcal{I}(\Sigma_1,\ldots,\Sigma_M)$ with a \emph{fully unknown topology}, composed of $M\in\mathbb{N}^+$ individual subsystems~$\Sigma_i$ with unknown dynamics. The primary goal is to develop a compositional data-driven scheme to construct storage functions $\mathcal S_i$ for unknown subsystems while providing a GAS certificate over interconnected dt-NS with correctness guarantees.
	\end{problem}
\end{resp}

\section{Data-Driven Construction of Storage Functions}\label{Data-Driven}

In the context of our work, we establish the structure of storage functions as
\begin{align}\label{Eq:3}
	\mathcal{S}_i(q_i,x_i)=\sum_{j=1}^{r} {q}_i^jp_i^j(x_i), 
\end{align}
with user-defined homogeneous basis functions $p_i^j(x_i)$ and unknown variables $q_i=[{q}_{i}^1;\ldots;q_i^r] \in \mathbb{R}^r$.

\begin{assumption}\label{Assm}
	We assume that the parametric storage function in~\eqref{Eq:3} is continuously differentiable and homogeneous of degree 2. Additionally, the unknown coefficients $q_i$ are assumed to lie within a compact set.
\end{assumption}

We begin by formulating the required conditions for constructing storage functions in Definition~\ref{Def:4} as the following robust convex program (RCP):

\begin{mini!}|s|[2]
	{[g_i;\mu_{R_i};\delta_{R_i}]}{\mu_{R_i}+ \delta_{R_i},}{\label{RCP}}\notag
	\addConstraint{\forall (x_i,w_i) \in \mathbb R^{n_i}\times \mathbb R^{p_i}\!\!: \Vert (x_i,w_i)\Vert = 1,}\notag
	\addConstraint{- \mathcal S_i(q_i,x_i)< \mu_{R_i},\label{RCP1}}
	\addConstraint{\mathcal S_i(q_i,f_i(x_i,w_i)) - \mathcal S_i(q_i,x_i)}{}\notag
	\addConstraint{-\begin{bmatrix}
			w_i\\
			x_i
		\end{bmatrix}^\top\!
		\begin{bmatrix}
			\mathcal X_i^{11}&\mathcal X_i^{12}\\
			\mathcal X_i^{21}&\mathcal X_i^{22}
		\end{bmatrix}\begin{bmatrix}
			w_i\\
			x_i
		\end{bmatrix}\leq \mu_{R_i},\label{RCP2}}
	\addConstraint{\begin{bmatrix}
			w_i\\
			x_i
		\end{bmatrix}^\top\!\!
		\begin{bmatrix}
			\mathcal X_i^{11}&\mathcal X_i^{12}\\
			\mathcal X_i^{21}&\mathcal X_i^{22}
		\end{bmatrix}\!\begin{bmatrix}
			w_i\\
			x_i
		\end{bmatrix} \leq \delta_{R_i}, \label{RCP3}}
	\addConstraint{g_i = [\mathcal X_i^{11};\mathcal X_i^{12};\mathcal X_i^{22};{q}_{i}^1;\dots;q_{i}^r],~ \!\mu_{R_i}\!,\delta_{R_i}\!\!\!\in\! \mathbb R.}\notag
\end{mini!}
To ensure the well-posedness of the RCP, we assume that all decision variables within vector $g_i$ are subject to compact constraints. We denote the optimal value of RCP by $\mu^*_{R_i} + \delta^*_{R_i}$. If $\mu^*_{R_i} \leq 0$, a solution to the RCP implies the fulfillment of conditions outlined in Definition~\ref{Def:4}, thus establishing $\mathcal S_i$ as a storage function for subsystems~\eqref{Eq:2}.

\begin{remark}
	It is worth noting that since condition~\eqref{RCP3} can be rewritten as
	\begin{align}
		\begin{bmatrix}
			w_i\\
			x_i
		\end{bmatrix}^\top\!\!\underbrace{\begin{bmatrix}
				\mathcal X_i^{11}&\mathcal X_i^{12}\\
				\mathcal X_i^{21}&\mathcal X_i^{22}
		\end{bmatrix}}_{\mathcal X_i}\!\begin{bmatrix}
			w_i\\
			x_i
		\end{bmatrix} \leq \delta_{R_i}\underbrace{\begin{bmatrix}
				w_i\\
				x_i
			\end{bmatrix}^\top\!\!\begin{bmatrix}
				w_i\\
				x_i
		\end{bmatrix}}_{\Vert (x_i,w_i)\Vert^2 = 1}, 
	\end{align}
	the satisfaction of this condition is equivalent to the satisfaction of $\mathcal X_i \leq  \delta_{R_i} \mathds{I}_{p_i+n_i}$. Since $\delta_{R_i}$ can take any real value, this condition does not impose restrictions on $\mathcal X_i$ when it is designed for different subsystems. Our compositional data-driven results provide a condition on the summation of $\delta_{R_i}$, ensuring that the supply rates of subsystems can influence one another and compensate for potential negative contributions within the interconnection topology (cf.~\eqref{Con2}).
\end{remark}

Addressing the proposed RCP in~\eqref{RCP} presents two significant challenges. Firstly, this RCP comprises infinitely many constraints due to the continuous nature of the state and internal input of the subsystem. Secondly, and more dauntingly, solving the RCP in~\eqref{RCP2} requires precise knowledge of the model $f_i$, which is unknown in our setting. These challenges motivated us to devise a data-driven approach for constructing storage functions, bypassing the direct solving of the RCP in~\eqref{RCP}. To accomplish this, we collect a set of two-consecutive sampled data from trajectories of unknown subsystems, represented as pairs $((\tilde x^z_i,\tilde w^z_i),f_i(\tilde x^z_i,\tilde w^z_i)), z\in\{1,\dots,\mathcal N_i\}$. In our proposed data-driven approach, the first step involves projecting all the data onto a unit sphere by normalizing them with respect to $\Vert (\tilde x^z_i,\tilde w^z_i)\Vert$, denoted as: 
\begin{align}\label{New6}
	((\hat x^z_i,\hat w^z_i),f_i(\hat x^z_i,\hat w^z_i)) = \frac{((\tilde x^z_i,\tilde w^z_i),f_i(\tilde x^z_i,\tilde w^z_i))}{\Vert (\tilde x^z_i,\tilde w^z_i)\Vert}.
\end{align}
According to Remark~\ref{Rem:10} and to ensure the well-posedness of our setting over the unit sphere, we consider $\tilde w_i^z=[{\tilde x_1^z;\ldots;\tilde x_{i-1}^z;\tilde x_{i+1}^z;\ldots;\tilde x_{M}^z}]$. Subsequently, we compute the maximum distance between any points on the unit sphere and the set of data points as follows:
\begin{align}\notag
	\varepsilon_{i} &= \max_{(x_i,w_i)}\min_z\Vert (x_i,w_i) - (\hat x^z_i,\hat w^z_i)\Vert,\\\label{EQ:41}
	&(x_i,w_i)\in \mathbb R^{n_i}\times \mathbb R^{p_i}\!: \Vert (x_i,w_i)\Vert = 1.
\end{align}

Taking into account $(\hat x^z_i,\hat w^z_i)\in \mathbb R^{n_i}\times \mathbb R^{p_i}$, where $z\in\{1,\dots,\mathcal N_i\}$, we introduce the following scenario convex program (SCP):

\begin{mini!}|s|[2]
	{[g_i;\mu_{\mathcal N_i};\delta_{\mathcal N_i}]}{\mu_{\mathcal N_i} + \delta_{\mathcal N_i},}{\label{SCP}}\notag
	\addConstraint{\forall(\hat x^z_i,\hat w^z_i)\in \mathbb R^{n_i}\!\times \mathbb R^{p_i}, \forall z\in\{1,\dots,\mathcal N_i\},}\notag
	\addConstraint{- \mathcal S_i(q_i,\hat x^z_i)< \mu_{\mathcal N_i},\label{SCP1}}
	\addConstraint{\mathcal S_i(q_i,f_i(\hat x^z_i,\hat w^z_i)) \!-\! \mathcal S_i(q_i,\hat x^z_i)}\notag
	\addConstraint{-\begin{bmatrix}
			\hat w^z_i\\
			\hat x^z_i
		\end{bmatrix}^\top\!\!
		\begin{bmatrix}
			\mathcal X_i^{11}&\mathcal X_i^{12}\\
			\mathcal X_i^{21}&\mathcal X_i^{22}
		\end{bmatrix}\!\begin{bmatrix}
			\hat w^z_i\\
			\hat x^z_i
		\end{bmatrix}\!\leq\! \mu_{\mathcal N_i},\label{SCP2}}
	\addConstraint{\begin{bmatrix}
			\hat w^z_i\\
			\hat x^z_i
		\end{bmatrix}^\top\!
		\begin{bmatrix}
			\mathcal X_i^{11}&\mathcal X_i^{12}\\
			\mathcal X_i^{21}&\mathcal X_i^{22}
		\end{bmatrix}\begin{bmatrix}
			\hat w^z_i\\
			\hat x^z_i
		\end{bmatrix} \!\leq \!\delta_{\mathcal N_i},\label{SCP3}}
	\addConstraint{g_i = [\mathcal X_i^{11};\mathcal X_i^{12};\mathcal X_i^{22};{q}_{i}^1;\dots;q_{i}^r],\!\!~\mu_{\mathcal N_i},\!\delta_{\mathcal N_i} \!\!\in\! \mathbb R.}\notag
\end{mini!}
It is evident that the proposed SCP in~\eqref{SCP} comprises a finite number of constraints of the same form as in~\eqref{RCP}. Let us denote the optimal value of SCP as $\mu_{\mathcal N_i}^* + \delta_{\mathcal N_i}^*$.

\begin{remark}
	While the unknown dynamics appear in the conditions of the SCP~\eqref{SCP}, they can be replaced using collected data. In fact, since the dynamics of each subsystem evolve recursively in discrete time, given an input-output data point, the input represents the pair of state and internal input (\emph{i.e.,} $(\hat x_i^z, \hat w_i^z)$), while the output corresponds to the unknown dynamics (\emph{i.e.,} $f(\hat x_i^z, \hat w_i^z)$).
\end{remark}

\begin{remark}
	Note that  $\mathcal{N}_{i}$ represents the number of trajectories with  arbitrary horizons, each of which allows the collection of one data pair $((\hat x_i^z, \hat w_i^z),f(\hat x_i^z, \hat w_i^z))$ from the $i$-th trajectory. 
\end{remark}

\begin{remark}
	Note that $\begin{bmatrix}
		\hat w^z_i\\
		\hat x^z_i
	\end{bmatrix}^\top\!\!\!\!
	\mathcal X_i^*\!\begin{bmatrix}
		\hat w^z_i\\
		\hat x^z_i
	\end{bmatrix}$ plays a key role in the compositionality condition we introduce in the following section, particularly in Theorem~\ref{Thm1}, which is based on data. This is the primary reason why we have incorporated an additional condition in~\eqref{RCP3} and~\eqref{SCP3}, where the objective of the optimization problems involves the minimization of both $\mu_{\mathcal N_i}$ and $\delta_{\mathcal N_i}$.
\end{remark}

In the next section, we solve the proposed SCP and construct the storage function of unknown $\Sigma_i$. Additionally, we formally verify the GAS property of unknown interconnected dt-NS by proposing some new compositionality conditions purely based on data. 

\section{GAS Certificate for Interconnected Network with Unknown Topology}\label{Guarantee_RCP}

Here, we assume that each unknown subsystem model $f_i$ is Lipschitz continuous with respect to $(x_i,w_i)$. Given that the storage function $\mathcal S_i$ is \emph{continuously differentiable} according to Assumption~\ref{Assm} and our analysis is on the unit sphere (bounded domain), one can readily conclude that $\mathcal S_i(q_i,f_i(x_i,w_i))^* - \mathcal S_i(q_i,x_i)^*$ is also Lipschitz continuous with respect to $(x_i,w_i)$ with a Lipschitz constant $\mathscr{L}^2_{i}$, where $\mathcal S_i(q_i,\cdot)^* = \mathcal S_i(q_i^*,\cdot)$. Similarly, under the same reasoning, one can readily show that $\mathcal S_i(q_i,x_i)^*$ is also Lipschitz continuous with respect to $x_i$, with a Lipschitz constant $\mathscr{L}^1_{i}$. Later on, we provide an algorithm to estimate Lipschitz constants $\mathscr{L}^1_{i},\mathscr{L}^2_{i}$ of unknown subsystems from data. We now propose the next theorem to offer formal GAS guarantee with correctness assurance over an unknown interconnected network with unknown topology.

\begin{theorem}\label{Thm1}
	Consider an interconnected dt-NS $\Sigma = \mathcal{I}(\Sigma_1,\ldots,\Sigma_M)$ composed of $M\in\mathbb{N}^+$ individual subsystems~$\Sigma_i$ with a fully unknown topology $\mathcal M$. Consider subsystems $\Sigma_i$ and their corresponding $\text{SCP}$ in~\eqref{SCP} with associated optimal value $\mu_{\mathcal N_i}^* + \delta_{\mathcal N_i}^*$ and optimal solutions $g_i^* = [\mathcal X_i^{11^*};\mathcal X_i^{12^*};\mathcal X_i^{22^*};{q}^*_{1_i};\dots;q^*_{r_i}]$ with $\mathcal N_i$ sampled data. If 
	\begin{subequations}
		\begin{align}\label{Con1}
			&\mu_{\mathcal N_i}^* \!+\! \mathscr{L}_{i}^1 \varepsilon_{i}  < 0, \quad \forall i\in\{1,\dots,M\},\\\label{Con2}
			&\sum_{i=1}^M\big(\mu_{\mathcal N_i}^* \!+\! \delta_{\mathcal N_i}^* \!+\! \mathscr{L}_{i}^2 \varepsilon_{i} \big) < 0,
		\end{align}
	\end{subequations}
	then  
	\begin{align}\label{newd}
		\mathcal V(q,x)= \sum_{i=1}^{M}  \mathcal S_i(q^*_i,x_i).
	\end{align}
	is a Lyapunov function for the interconnected system $\Sigma = \mathcal{I}(\Sigma_1,\ldots,\Sigma_M)$ satisfying conditions~\eqref{alpha0}-\eqref{alpha1}. Consequently, unknown interconnected dt-NS $\Sigma$ is GAS with correctness guarantees, according to Definition~\ref{GAS}.
\end{theorem}

\begin{proof}
	We first show that under condition~\eqref{Con1}, Lyapunov function $\mathcal V$ in~\eqref{Lyapunov} satisfies condition~\eqref{alpha0} over the unit sphere,
	\emph{i.e.,}
	\begin{align}\label{new2}
		\mathcal V(q,x) >0, \quad \quad {\forall x\in \mathbb R^n \backslash \{0 \}\!: \Vert x\Vert = 1.}
	\end{align}
	We write down $\mathcal V$ based on storage functions of individual subsystems as:
	\begin{align*}
		&-\mathcal V(q,x)= - \sum_{i=1}^{M}  \mathcal S_i(q^*_i,x_i).
	\end{align*}
	For given $x_i$ and with a slight abuse of notation, let
	$\bar z:= \arg \min_z \Vert x_i - \hat x^z_i\Vert.$ By incorporating the terms $\sum_{i=1}^{M}\mathcal S_i(q_i,\hat x^{\bar z}_i)^*$ through addition and subtraction, with $\mathcal S_i(q_i,\cdot)^* = \mathcal S_i(q_i^*,\cdot),$ one has
	\begin{align*}
		- &\mathcal V(x,q)= \sum_{i=1}^{M} \big(\!- \mathcal S_i(q^*_i,x_i) \!+\! \mathcal S_i(q_i,\hat x^{\bar z}_i)^*\!-\! \mathcal S_i(q_i,\hat x^{\bar z}_i)^*\big).
	\end{align*}
	Given that $\mathcal S_i(q_i,x_i)^*$ is Lipschitz continuous with respect to $x_i$ with a Lipschitz constant $\mathscr{L}^1_{i}$, and since
	\begin{align*}
		\underset{x_i: ~\!\|x_i\| \leq 1}\max \| x_i - \hat{x}_i^z \| &= \underset{(x_i,w_i): ~\!\|(x_i,w_i)\|=1}\max \| x_i - \hat{x}_i^z \|, \\
		\| x_i - \hat{x}_i^z \| &\leq \| (x_i,w_i) - (\hat{x}_i^z ,\hat{w}_i^z) \|,
	\end{align*}
	we have
	\begin{align*}
		- \mathcal V(q,x)&\leq \sum_{i=1}^{M} \big(\mathscr{L}^1_{i}\min_z\Vert x_i - \hat x^z_i\Vert\!-\! \mathcal S_i(q_i,\hat x^{\bar z}_i)^*\big)\\ 
		&\leq \sum_{i=1}^{M}\!\big(\mathscr{L}^1_{i}\!\max_{x_i: ~\!\|x_i\| \leq 1}\!\min_z\!\Vert x_i-\hat x^z_i \Vert \!-\! \mathcal S_i(q_i,\hat x^{\bar z}_i)^*\big)\\
		& \leq \sum_{i=1}^{M}\big(\mathscr{L}_{i}^1 \varepsilon_i - \mathcal S_i(q_i,\hat x^{\bar z}_i)^*\big).
	\end{align*}
	According to condition~\eqref{SCP1} of SCP, we have
	\begin{align}\label{new90}
		- \mathcal V(q,x)\leq \sum_{i=1}^{M}\big(\mu_{\mathcal N_i}^* + \mathscr{L}^1_{i} \varepsilon_{i}\big).
	\end{align}
	Given our proposed condition in~\eqref{Con1}, one can conclude that
	\begin{align*}
		\mathcal V(q,x)> 0,  \quad {\forall x\in \mathbb R^n \backslash \{0 \}\!: \Vert x\Vert = 1.}
	\end{align*}
	Given that $\mathcal V$ in \eqref{newd} is based on the linear combination of $\mathcal S_i({q_i},x_i)$ and according to the inequality obtained in~\eqref{new90}, $\mu_{R_i}^*$ in RCP \eqref{RCP} can be defined as $\mu_{R_i}^*= \mu_{\mathcal N_i}^* + \mathscr{L}_{i}^1 \varepsilon_{i}$. This expression is negative for all subsystems $M$ under the condition in~\eqref{Con1}.
	
	We proceed with showing that under condition~\eqref{Con2}, $\mathcal V$ in the form of~\eqref{Lyapunov} satisfies condition~\eqref{alpha1}, as well, over the unit sphere, \emph{i.e.,}
	\begin{align*}
		\mathcal V(q,f(x)) - \mathcal V(q,x) < 0, \quad \quad \quad {\forall x\in \mathbb R^n\!\!: \Vert x\Vert = 1.}
	\end{align*}
	Given the proposed Lyapunov function in~\eqref{Lyapunov}, one has
	\begin{align*}
		\mathcal V(q,f(x)) - \mathcal V(q,x) = \sum_{i=1}^{M}\big(\mathcal S_i(q^*_i,f_i(x_i,w_i)) -  \mathcal S_i(q^*_i,x_i)\big)
	\end{align*}
	which is well-posed over the unit sphere since $\Vert x\Vert  =1 \Leftrightarrow  \Vert (x_i, w_i)\Vert =1$ (cf. Remark~\ref{Rem:10}).
	For given $(x_i,w_i)$ and with a slight abuse of notation, let
	$\bar z:= \arg \min_z \Vert (x_i,w_i) - (\hat x_i^{z},\hat w_i^{z}) \Vert$.  By incorporating the terms $\sum_{i=1}^{M}(\mathcal S_i(q_i,f_i(q_i,\hat x_i^{\bar z},\hat w_i^{\bar z}))^* - \mathcal S_i(q_i,\hat x_i^{\bar z})^*)$ through addition and subtraction, one has
	\begin{align*}
		\mathcal V(q,f(x)) - \mathcal V(q,x) &= \sum_{i=1}^{M}\big(\mathcal S_i(q^*_i,f_i(x_i,w_i)) -  \mathcal S_i(q^*_i,x_i)\\
		& -(\mathcal S_i(q_i,f_i(\hat x^{\bar z}_i,\hat w^{\bar z}_i))^* - \mathcal S_i(q_i,\hat x^{\bar z}_i)^* )\\
		&+~ \mathcal S_i(q_i,f_i(\hat x^{\bar z}_i,\hat w^{\bar z}_i))^* - \mathcal S_i(q_i,\hat x^{\bar z}_i)^*\big).
	\end{align*}
	Given that $\mathcal S_i(q_i,f_i(x_i,w_i))^* - \mathcal S_i(q_i,x_i)^*$ is Lipschitz continuous with respect to $(x_i,w_i)$ with a Lipschitz constant $\mathscr{L}_{i}^2$, we have
	\begin{align*}
		\mathcal V&(q,f(x)) - \mathcal V(q,x) \\
		&\leq \sum_{i=1}^{M}\big(\mathscr{L}_{i}^2 \, \min_z\Vert (x_i,w_i) - (\hat x^{z}_i,\hat w^{z}_i)\Vert\\
		&~~~+ \mathcal S_i(q_i,f_i(\hat x^{\bar z}_i,\hat w^{\bar z}_i))^* - \mathcal S_i(q_i,\hat x^{\bar z}_i)^*\big)\\
		&\leq \mathscr{L}_{i}^2 \, \max_{(x_i,w_i)}\min_z\Vert (x_i,w_i) - (\hat x^{z}_i,\hat w^{z}_i)\Vert\\
		&~~~+ \mathcal S_i(q_i,f_i(\hat x^{\bar z}_i,\hat w^{\bar z}_i))^* - \mathcal S_i(q_i,\hat x^{\bar z}_i)^*\big)\\
		& = \sum_{i=1}^{M}\big(\mathscr{L}_{i}^2 \varepsilon_i + \mathcal S_i(q_i,f_i(\hat x^{\bar z}_i,\hat w^{\bar z}_i))^* - \mathcal S_i(q_i,\hat x^{\bar z}_i)^*\big).
	\end{align*}
	According to condition~\eqref{SCP2} of SCP, we have
	\begin{align*}
		\mathcal V(q,f(x)) - \mathcal V(q,x)&\leq \sum_{i=1}^M\big(\mu_{\mathcal N_i}^* + \delta_{\mathcal N_i}^* + \mathscr{L}_{i}^2 \varepsilon_{i} \big).
	\end{align*}
	Given our proposed condition in~\eqref{Con2}, one can conclude that
	\begin{align*}
		\mathcal V(q,f(x)) - \mathcal V(q,x) < 0, \quad \quad {\forall x\in \mathbb R^n\!\!: \Vert x \Vert = 1.}
	\end{align*}
	Then the unknown interconnected network is GAS with correctness guarantees over the unit sphere. Given the homogeneity property of interconnected network and the constructed Lyapunov function, the GAS certificate can then be globally guaranteed in $\mathbb R^{n}$, which concludes the proof.
\end{proof}

\begin{remark}
	The proposed conditions~\eqref{Con1}-\eqref{Con2} introduce a novel notion of compositionality in the dissipativity setting using data, which are proposed in our work for the first time. Given that $\hat w^z_i$ captures the interaction effect between subsystems in the unknown interconnected topology, condition~\eqref{Con2} allows supply rates of subsystems to be designed in a way that they can compensate the potential undesirable effect of neighboring subsystems while this condition becomes negative. Note that compositionality conditions~\eqref{Con1}-\eqref{Con2} do not need to be verified if $\mu_{\mathcal N_i}^* + \mathscr{L}^1_{i} \varepsilon_{i}$ and $\mu_{\mathcal N_i}^* + \delta_{\mathcal N_i}^* + \mathscr{L}^2_{i} \varepsilon_{i}$ can become negative while solving SCP for each individual subsystem. In this case, the results of the current work can be extended to the stability analysis of unknown networks with an \emph{arbitrary, a-priori undefined number of subsystems}.
\end{remark}

\begin{algorithm}[t]
	\caption{Estimation of $\mathscr{L}^1_i, \mathscr{L}^2_i$ via data}
	\label{Alg:1}		
	\begin{center}
		\begin{algorithmic}[1]
			\REQUIRE $\mathcal S_i(q_i,\hat x^z_i)^*, \mathcal S_i(q_i,f_i(\hat x^z_i,\hat w^z_i))^*$
			\STATE Choose $\rho, \sigma \in \mathbb N^+$ and $\alpha \in \mathbb R^+$
			\STATE {\bf for} $i = 1\!:\!\rho$
			\STATE ~~~~~Select sampled pairs $(\hat x_{i}^z,\hat w_{i}^z), (\hat x_{i}^{z'},\hat w_{i}^{z'})$ from\\ ~~~~ $\mathbb R^{n_i}\times \mathbb R^{p_i}$ such that $\Vert (\hat x_{i}^z,\hat w_{i}^z) - (\hat x_{i}^{z'},\hat w_{i}^{z'}) \Vert \leq \alpha$ 
			\STATE ~~~~~Compute the slope $\theta_i$ as $$ \!\!\!\theta_i = \frac{\Vert h_i(\hat x_{i}^z,\hat w_{i}^z) -h_i(\hat x_{i}^{z'},\hat w_{i}^{z'}) \Vert}{\Vert (\hat x_{i}^z,\hat w_{i}^z) - (\hat x_{i}^{z'},\hat w_{i}^{z'}) \Vert}$$ ~~~~~with $h_i(\hat x_{i}^z,\hat w_{i}^z) = \mathcal S_i(q_i,f_i(\hat x^z_i,\hat w^z_i))^* \!-\! \mathcal S_i(q_i,\hat x^z_i)^*$
			\STATE {\bf end}
			\STATE Compute the maximum slope as $$\Theta = \max \{\theta_1,\dots,\theta_{\rho}\}$$
			\STATE Repeat Steps $2$-$6$, $\sigma$ times and acquire $\Theta_1, \dots, \Theta_\sigma$
			\STATE Apply \emph{Reverse Weibull distribution}~\cite{wood1996estimation} to $\Theta_1, \dots, \Theta_\sigma$ which provides us with, so-called, location, scale, and shape parameters
			\STATE The obtained \emph{location parameter} is the estimated $\mathscr{L}^2_{i}$
			\STATE Repeat Steps $1$-$9$ with $h_i = \mathcal S_i(q_i,\hat x^{z}_i)^*$ to estimate $\mathscr{L}_{i}^1$
			\ENSURE $\mathscr{L}^1_i, \mathscr{L}^2_i$
		\end{algorithmic}
	\end{center}
\end{algorithm}	

In order to verify conditions~\eqref{Con1}-\eqref{Con2} in Theorem~\ref{Thm1}, the computation of $\mathscr{L}^1_i, \mathscr{L}^2_i$ is a prerequisite. To achieve this, we introduce Algorithm~\ref{Alg:1}, which enables the estimation of $\mathscr{L}^1_i, \mathscr{L}^2_i$ using a finite set of data. Within this algorithm, we rely on Lemma~\ref{lemma}, adopted from~\cite{wood1996estimation}, to guarantee the convergence of the estimated $\mathscr{L}^1_i, \mathscr{L}^2_i$ towards their actual values.

\begin{lemma}\label{lemma}
	Under Algorithm~\ref{Alg:1}, the estimated $\mathscr{L}^1_i, \mathscr{L}^2_i$ converge to their actual values if and only if $\alpha$ approaches zero, and both $\rho$ and $\sigma$ go to infinity.
\end{lemma}

\begin{remark}
	While the results of Lemma~\ref{lemma} provide estimation guarantees in the limit, the amount of data required for Lipschitz constant estimation in practice does not necessarily need to approach infinity. Specifically, in some scenarios, we observed that the estimated Lipschitz constant remains unchanged after a certain number of data points, as it matches the exact Lipschitz constant derived from the mathematical model.
\end{remark}

\begin{remark}
	Given that unknown coefficients $q_i$ are required for estimating the Lipschitz constants $\mathscr{L}^1_i, \mathscr{L}^2_i$ in Algorithm~\ref{Alg:1}, and subsequently verifying conditions~\eqref{Con1}-\eqref{Con2}, it is imperative to first solve the proposed SCP in~\eqref{SCP}. To avoid a posteriori verification of conditions~\eqref{Con1}-\eqref{Con2}, one can initially consider unknown coefficients of $q_i$ to exist within a certain range and estimate the Lipschitz constants $\mathscr{L}^1_i$ and $\mathscr{L}^2_i$ before solving the SCP. Subsequently, the underlying ranges should be enforced during the solution of SCP. For instance, one could assume $-2 \leq q_i \leq 2$ and estimate the Lipschitz constants $\mathscr L_i^1$ and $\mathscr L_i^2$ based on any $q_i$ within that range. Then, when solving the SCP in~\eqref{SCP}, this range should be enforced as an additional constraint to ensure that $q_i$ remains within the specified range. This ensures that the a-priori computation of the Lipschitz constants based on unknown $q_i$ is valid.
\end{remark}

\begin{remark}
	Note that while the interconnection topology is unknown,  the full-state measurement of the subsystems within $\hat w_i^z$ should be used for estimating the Lipschitz constants $\mathscr{L}^1_i$ and $\mathscr{L}^2_i$, \emph{i.e.,} $\hat w_i^z=[{\hat x_1^z;\ldots;\hat x_{i-1}^z;\hat x_{i+1}^z;\ldots;\hat x_{M}^z}]$ (cf. Remark~\ref{Rem:10}).
\end{remark}

We now propose Algorithm~\ref{Alg:2} to describe the required steps for ensuring the GAS certificate over unknown interconnected dt-NS with an unknown topology.

\begin{remark}
	The proposed data-driven findings can be generalized by assuming that each unknown subsystem is noisy, treating the stochasticity as part of the unknown dynamics. In this case, one can explore the probabilistic stability of the system by introducing an expected value in the second condition of the Lyapunov function in~\eqref{alpha1} and the storage function in~\eqref{Eq:8_2} for the individual subsystems. Since the expected value must be first approximated empirically using Chebyshev's inequality~\cite{saw1984chebyshev} and then used in the SCP in~\eqref{SCP}, the final guarantee comes with a confidence level that has a closed-form relationship with the required number of noise realizations for the empirical computation.
\end{remark}

\begin{remark}
	Note that Algorithm~\ref{Alg:1} is specifically for the estimation of the Lipschitz constant and always terminates after completing the required steps. Algorithm~\ref{Alg:2} is a \emph{pseudocode} outlining the required steps for verifying the GAS property over the interconnected network. For instance, if the network is GAS, the algorithm will break and return this outcome. However, Step 6 of the algorithm is a refinement process aimed at verifying the GAS property of the network using some potential refinements. It is important to note that since Lyapunov functions provide sufficient (not necessary) conditions for nonlinear systems, if Algorithm~\ref{Alg:2} cannot successfully terminated, this does not imply that the network is unstable.
\end{remark}

\begin{algorithm}[t]
	\caption{GAS certificate of interconnected dt-NS with fully unknown topology}
	\label{Alg:2}		
	\begin{center}
		\begin{algorithmic}[1]
			\REQUIRE Collect data-points $((\tilde x^z_i,\tilde w^z_i),f_i(\tilde x^z_i,\tilde w^z_i)), z\in\{1,\dots,\mathcal N_i\}$
			\STATE Project all data-points onto unit sphere by normalizing them as $((\hat x^z_i,\hat w^z_i),f_i(\hat x^z_i,\hat w^z_i)) = \frac{((\tilde x^z_i,\tilde w^z_i),f_i(\tilde x^z_i,\tilde w^z_i))}{\Vert (\tilde x^z_i,\tilde w^z_i)\Vert}$
			\STATE Compute Lipschitz constants $\mathscr{L}^1_i, \mathscr{L}^2_i$ according to Algorithm~\ref{Alg:1} 
			\STATE Compute $\varepsilon_{i}$ as the maximum distance between any points in the unit sphere and the set of data points, \emph{i.e.,} $
			\varepsilon_{i} = \underset{(x_i,w_i)}\max\underset{z}\min\Vert (x_i,w_i) - (\hat x^z_i,\hat w^z_i)\Vert,~ \forall (x_i,w_i)\in \mathbb R^{n_i}\times \mathbb R^{p_i}\!: \Vert (x_i,w_i)\Vert = 1$
			\STATE Solve $\text{SCP}$ in~\eqref{SCP} with the normalized data and obtain $\mu^*_{\mathcal N_i}$ and $\delta_{\mathcal N_i}^*$ 
			\STATE If $\mu_{\mathcal N_i}^* + \mathscr{L}^1_{i} \varepsilon_{i}  < 0, \forall i\in\{1,\dots,M\},$ and $\sum_{i=1}^M\big(\mu_{\mathcal N_i}^* + \delta_{\mathcal N_i}^* + \mathscr{L}_{i}^2 \varepsilon_{i}  \big) < 0$, then unknown interconnected dt-NS is GAS with Lyapunov function $\mathcal V(q,x) := \sum_{i=1}^{M}\mathcal S_i(q_i,x_i)$
			\STATE Otherwise, repeat Steps 2-5 with additional collected data $\mathcal N_i$ to reduce $\varepsilon_{i}$, or increase the terms of the homogeneous basis functions, potentially designing other $\mu^*_{\mathcal N_i}$ and $\delta_{\mathcal N_i}^*$ that satisfy conditions~\eqref{Con1}-\eqref{Con2}
		\end{algorithmic}
	\end{center}
\end{algorithm}

\section{Unknown Interconnected Networks with Known topology}\label{ADMM}

In the scenario where subsystems are unknown but the interconnection topology $\mathcal M$ is known, we provide here our sub-result where we initially construct a storage function for each unknown subsystem based on data. This is achieved by solving an SCP for each unknown subsystem, ensuring a guaranteed confidence level of 1 over the SF construction. In this case, we show that the data-driven conditions~\eqref{Con1}-\eqref{Con2} are reduced to $\mu_{\mathcal N_i}^* \!+\! \mathscr{L}_{i} \varepsilon_{i} < 0$ with $\mathscr{L}_{i} = \max\{\mathscr{L}^1_{i},\mathscr{L}^2_{i}\}$ in the level of subsystems. Subsequently, we examine the traditional dissipativity compositionality condition~\eqref{EQQ:2} to determine its satisfaction. To avoid potential posteriori checks for this condition, we employ an alternating direction method of multipliers (ADMM) algorithm~\cite{boyd2011distributed}, enabling efficient satisfaction of the compositionality condition in a \emph{distributed manner} while searching for storage functions of individual subsystems. 

We initially propose a relaxed alternative scenario convex program (SCP) that does not necessitate condition~\eqref{SCP3}, outlined as follows:

\begin{mini!}|s|[2]
	{[g_i;\mu_{\mathcal N_i}]}{\mu_{\mathcal N_i},}{\label{SCP-1}}\notag
	\addConstraint{\forall(\hat x^z_i,\hat w^z_i)\in \mathbb R^{n_i}\!\times \mathbb R^{p_i}, \forall z\in\{1,\dots,\mathcal N_i\},}\notag
	\addConstraint{- \mathcal S_i(q_i,\hat x^z_i)< \mu_{\mathcal N_i},\label{SCP1-1}}
	\addConstraint{\mathcal S_i(q_i,f_i(\hat x^z_i,\hat w^z_i)) \!-\! \mathcal S_i(q_i,\hat x^z_i)}\notag
	\addConstraint{-\begin{bmatrix}
			\hat w^z_i\\
			\hat x^z_i
		\end{bmatrix}^\top\!
		\begin{bmatrix}
			\mathcal X_i^{11}&\mathcal X_i^{12}\\
			\mathcal X_i^{21}&\mathcal X_i^{22}
		\end{bmatrix}\begin{bmatrix}
			\hat w^z_i\\
			\hat x^z_i
		\end{bmatrix}\leq \mu_{\mathcal N_i},\label{SCP2-1}}
	\addConstraint{g_i = [\mathcal X_i^{11};\mathcal X_i^{12};\mathcal X_i^{22};{q}_{i}^1;\dots;q_{i}^r],~ \mu_{\mathcal N_i} \in \mathbb R.}\notag
\end{mini!}

\begin{remark}
	Note that in this section, it is not required to consider the entire state measurement of subsystems within 
	$w_i$, as described in Remark~\ref{Rem:10}, during the analysis. In fact, the primary objective here is to demonstrate that each subsystem individually exhibits dissipativity \emph{over the unit sphere} with respect to a supply rate via a storage function $\mathcal{S}_i$, as specified in the SCP~\eqref{SCP-1}. Subsequently, this result can be extended globally to $\mathbb{R}^{n_i}$ by leveraging the \emph{homogeneity property} of the maps $f_i$ and storage functions $\mathcal{S}_i$. Ultimately, using the dissipativity compositional condition~~\eqref{EQQ:2} and incorporating the knowledge of the interconnection matrix $\mathcal{M}$, the GAS certificate is ensured for the entire network over $\mathbb{R}^n$.
\end{remark}

We now raise the following theorem, enabling the construction of a storage function for each unknown subsystem through solving the SCP in~\eqref{SCP-1}, while ensuring guaranteed confidence based on collected data.

\begin{theorem}\label{Thm2}
	Consider individual subsystems in~\eqref{Eq:2} and their corresponding $\text{SCP}$ in~\eqref{SCP-1} with associated optimal value $\mu_{\mathcal N_i}^*$ and solutions $g_i^* = [\mathcal X_i^{11^*};\mathcal X_i^{12^*};\mathcal X_i^{22^*};{q}^*_{1_i};\dots;q^*_{r_i}]$ with $\mathcal N_i$ sampled data. If 
	\begin{align}\label{Con3}
		\mu_{\mathcal N_i}^* \!+\! \mathscr{L}_{i} \varepsilon_{i} < 0,
	\end{align}
	with $\mathscr{L}_{i} = \max\{\mathscr{L}^1_{i},\mathscr{L}^2_{i}\}$, then the constructed $\mathcal{S}_i$ obtained by solving SCP~\eqref{SCP-1} serve as storage functions for unknown subsystems $\Sigma_i$ globally over $\mathbb{R}^{n_i}$, with  guarantees of correctness.
\end{theorem}

\begin{proof}
	Initially, we demonstrate that the constructed $\mathcal{S}_i$ obtained by solving SCP~\eqref{SCP-1} satisfy condition~\eqref{Eq:8_2} across the entire state and disturbance spaces, with $\mu_{\mathcal N_i} \leq 0$, \emph{i.e.,}
	\begin{align}\label{N12}
		&\mathcal S_i(q^*_i,f_i(x_i,w_i)) - \mathcal S_i(q^*_i,x_i) - \begin{bmatrix}
			w_i\\
			x_i
		\end{bmatrix}^\top\!\!\!
		\mathcal X_i^*\begin{bmatrix}
			w_i\\
			x_i
		\end{bmatrix}\leq 0,\\\notag
		&\forall (x_i,w_i) \!\in\!  \mathbb R^{n_i}\!\times \mathbb R^{p_i}\!: \Vert (x_i,w_i)\Vert = 1.
	\end{align}
	Given the Lipschitz continuity of $f_i(x_i,w_i)$, it is readily demonstrated that the left-hand side of~\eqref{N12} remains Lipschitz continuous with respect to $(x_i,w_i)$, with a Lipschitz constant $\mathscr{L}_{i}^2 \in\mathbb{R}^+$ (with a slight abuse of notation), since the storage function $\mathcal{S}_i$ is continuously differentiable, and our analysis is conducted on the unit sphere (bounded domain). 
	For given $(x_i,w_i)$ and with a slight abuse of notation, let
	$\bar{z}:= \arg \min_z \Vert (x_i,w_i) - (\hat x^z_i,\hat w^z_i) \Vert$.  By incorporating the terms $\mathcal S_i(q_i,f_i(q_i,\hat x^{\bar{z}}_i,\hat w^{\bar{z}}_i))^* - \mathcal S_i(q_i,\hat x^{\bar{z}}_i)^* - \begin{bmatrix}
		\hat w^{\bar{z}}_i\\
		\hat x^{\bar{z}}_i
	\end{bmatrix}^\top\!\!\!\!
	\mathcal X_i^*\begin{bmatrix}
		\hat w^{\bar{z}}_i\\
		\hat x^{\bar{z}}_i
	\end{bmatrix}$ through addition and subtraction, one has
	\begin{align*}
		&\mathcal S_i(q^*_i,f_i(x_i,w_i)) - \mathcal S_i(q^*_i,x_i) - \begin{bmatrix}
			w_i\\
			x_i
		\end{bmatrix}^\top\!\!\!
		\mathcal X_i^*\begin{bmatrix}
			w_i\\
			x_i
		\end{bmatrix}\\
		& = \mathcal S_i(q^*_i,f_i(x_i,w_i)) - \mathcal S_i(q^*_i,x_i) - \begin{bmatrix}
			w_i\\
			x_i
		\end{bmatrix}^\top\!\!\!
		\mathcal X_i^*\begin{bmatrix}
			w_i\\
			x_i
		\end{bmatrix}\\
		& ~~~-( \mathcal S_i(q_i,f_i(q_i,\hat x^{\bar{z}}_i,\hat w^{\bar{z}}_i))^* - \mathcal S_i(q_i,\hat x^{\bar{z}}_i)^* - \begin{bmatrix}
			\hat w^{\bar{z}}_i\\
			\hat x^{\bar{z}}_i
		\end{bmatrix}^\top\!\!\!\!
		\mathcal X_i^*\begin{bmatrix}
			\hat w^{\bar{z}}_i\\
			\hat x^{\bar{z}}_i
		\end{bmatrix})\\
		&~~~~+\mathcal S_i(q_i,f_i(q_i,\hat x^{\bar{z}}_i,\hat w^{\bar{z}}_i))^* - \mathcal S_i(q_i,\hat x^{\bar{z}}_i)^* - \begin{bmatrix}
			\hat w^{\bar{z}}_i\\
			\hat x^{\bar{z}}_i
		\end{bmatrix}^\top\!\!\!\!
		\mathcal X_i^*\begin{bmatrix}
			\hat w^{\bar{z}}_i\\
			\hat x^{\bar{z}}_i
		\end{bmatrix}\!\!,
	\end{align*}
	with $\mathcal S_i(q_i,\cdot)^* = \mathcal S_i(q_i^*,\cdot) .$
	Given the Lipschitz continuity of~\eqref{N12} with a Lipschitz constant $\mathscr{L}_{i}^2$, we have
	\begin{align*}
		&\mathcal S_i(q^*_i,f_i(x_i,w_i)) - \mathcal S_i(q^*_i,x_i) - \begin{bmatrix}
			w_i\\
			x_i
		\end{bmatrix}^\top\!\!\!
		\mathcal X_i^*\begin{bmatrix}
			w_i\\
			x_i
		\end{bmatrix}\\
		&\leq \mathscr{L}_{i}^2 \, \min_z\Vert (x_i,w_i) - (\hat x^{{z}}_i,\hat w^{{z}}_i)\Vert\\
		& ~~~+ \mathcal S_i(q_i,f_i(q_i,\hat x^{\bar{z}}_i,\hat w^{\bar{z}}_i))^* - \mathcal S_i(q_i,\hat x^{\bar{z}}_i)^* - \begin{bmatrix}
			\hat w^{\bar{z}}_i\\
			\hat x^{\bar{z}}_i
		\end{bmatrix}^\top\!\!\!\!
		\mathcal X_i^*\begin{bmatrix}
			\hat w^{\bar{z}}_i\\
			\hat x^{\bar{z}}_i
		\end{bmatrix}\\
		&\leq \mathscr{L}_{i}^2 \, \max_{(x_i,w_i)}\min_z\Vert (x_i,w_i) - (\hat x^z_i,\hat w^z_i)\Vert\\
		& ~~~+ \mathcal S_i(q_i,f_i(q_i,\hat x^{\bar{z}}_i,\hat w^{\bar{z}}_i))^* - \mathcal S_i(q_i,\hat x^{\bar{z}}_i)^* - \begin{bmatrix}
			\hat w^{\bar{z}}_i\\
			\hat x^{\bar{z}}_i
		\end{bmatrix}^\top\!\!\!\!
		\mathcal X_i^*\begin{bmatrix}
			\hat w^{\bar{z}}_i\\
			\hat x^{\bar{z}}_i
		\end{bmatrix}\\
		& = \mathscr{L}_{i}^2 \varepsilon_i \!+ \mathcal S_i(q_i,\!f_i(q_i,\!\hat x^{\bar{z}}_i,\!\hat w^{\bar{z}}_i)\!)^* \!-\! \mathcal S_i(q_i,\hat x^{\bar{z}}_i)^* \!-\! \begin{bmatrix}
			\hat w^{\bar{z}}_i\\
			\hat x^{\bar{z}}_i
		\end{bmatrix}^\top\!\!\!\!\!
		\mathcal X_i^*\!\!\begin{bmatrix}
			\hat w^{\bar{z}}_i\\
			\hat x^{\bar{z}}_i
		\end{bmatrix}\!\!.
	\end{align*}
	According to condition~\eqref{SCP2-1} of SCP, we have
	\begin{align*}
		&\mathcal S_i(q^*_i,f_i(x_i,w_i)) - \mathcal S_i(q^*_i,x_i) - \begin{bmatrix}
			w_i\\
			x_i
		\end{bmatrix}^\top\!\!\!
		\mathcal X_i^*\begin{bmatrix}
			w_i\\
			x_i
		\end{bmatrix}\\
		&\leq \mu_{\mathcal N_i}^* + \mathscr{L}_{i} \varepsilon_{i}.
	\end{align*}
	Given our proposed condition in~\eqref{Con3}, one can conclude that
	\begin{align*}
		&\mathcal S_i(q^*_i,f_i(x_i,w_i)) - \mathcal S_i(q^*_i,x_i) - \begin{bmatrix}
			w_i\\
			x_i
		\end{bmatrix}^\top\!\!\!
		\mathcal X_i^*\begin{bmatrix}
			w_i\\
			x_i
		\end{bmatrix}\leq 0\\\notag
		&\forall (x_i,w_i) \!\in\!  \mathbb R^{n_i}\!\times \mathbb R^{p_i}\!: \Vert (x_i,w_i)\Vert = 1.
	\end{align*}
	This result can subsequently be extended globally to $\mathbb{R}^{n_i}\times \mathbb{R}^{p_i}$, thanks to the \emph{homogeneity property} of maps $f_i$ and storage functions $\mathcal{S}_i$. By employing a similar reasoning, it can be demonstrated that under condition~\eqref{Con3}, the $\mathcal{S}_i$ constructed through solving SCP in~\eqref{SCP-1} satisfies $\mathcal{S}_i(q_i^*,x_i) > 0$ across the entire state space. Then the $\mathcal{S}_i$ obtained by solving SCP in~\eqref{SCP-1} serve as storage functions for unknown subsystems $\Sigma_i, i \in \{1, \dots, M\}$, globally over $\mathbb{R}^{n_i}$, with guarantees of correctness, thereby concluding the proof.
\end{proof}

In order to construct the Lyapunov function for the interconnected system based on storage functions of individual subsystems derived from data, one needs to satisfy the traditional compositionality condition in~\eqref{EQQ:2}. Here, we employ the ADMM algorithm~\cite{boyd2011distributed} to efficiently satisfy condition~\eqref{EQQ:2} by decomposing it into local sub-problems. To do so, we define the following local
$\mathcal{G}_i\!=\! \big\{\mathcal X_i \, \big |  \mathcal{X}_i= \mathcal{X}_i^\top \big\}.$
In addition, we define the global constraint as:
\begin{align}
	\mathcal{G}\!=\! \Big\{
	(\mathcal X_1, \ldots, \mathcal X_M)
	\, \big | \, ~\text{condition}~\eqref{EQQ:2}~ \text{holds}\Big\}.
\end{align}
We now recast applicability of Theorem \ref{Thm:3} as a feasibility problem in the following lemma. 

\begin{lemma} \label{lem:admm}
	Consider an interconnected dt-NS $\Sigma = \mathcal{I}(\Sigma_1,\ldots,\Sigma_M)$. If there exist matrices $\mathcal X_1, \ldots, \mathcal X_M$ such that
	\begin{align} \notag
		&\mathcal X_i\in \mathcal{G}_i, \quad \forall i \in \{1,\ldots,M\}, \\\label{eq:opt}
		&(\mathcal X_1, \ldots, \mathcal X_M) \in \mathcal{G},
	\end{align}
	then $\mathcal V(x) := \sum_{i=1}^{M}\mathcal S_i(x_i)$ is a Lyapunov function for the interconnected system $\Sigma = \mathcal{I}(\Sigma_1,\ldots,\Sigma_M)$.
\end{lemma}

To reformulate the feasibility problem in Lemma~\ref{lem:admm} as a standard ADMM problem, we first define the following indicator functions:

\begin{align*}
	\mathds{1}_{\mathcal{G}_i}(\mathcal X_i)= 
	&\begin{cases}
		0,&\mathcal X_i \in \mathcal{G}_i,  \\
		+\infty, & \text{otherwise},
	\end{cases} \\
	\mathds{1}_\mathcal{G}(\mathcal X_1,\ldots,\mathcal X_M)=
	&\begin{cases}
		0, &(\mathcal X_1,\ldots, \mathcal X_M) \in \mathcal{G},  \\
		+\infty, & \text{otherwise}.
	\end{cases}
\end{align*}
Now by introducing auxiliary variables $\mathcal Z_i\in \mathbb R^{(n_i+p_i)\times (n_i+p_i)}$ for each subsystem, we rewrite~\eqref{eq:opt} as an optimization problem in the ADMM form as:
\begin{align}\label{eq:admm}
	&\text{ADMM}\!:\left\{
	\hspace{-0.5mm}\begin{array}{l}\min\limits_{\mathcal X_i,\mathcal Z_i} \quad\!\!\sum_{i = 1}^M(\mathds{1}_{\mathcal{G}_i}(\mathcal X_i))+\mathds{1}_\mathcal{G}(\mathcal Z_1,\ldots, \mathcal Z_M),\\
		\, \text{s.t.} \quad \quad \mathcal X_i- \mathcal Z_i = 0, \quad \forall i \in \{1,\ldots,M\}.\end{array}\right.
\end{align}
Since the first part of the objective function in~\eqref{eq:admm}, \emph{i.e.,} $\mathds{1}_{\mathcal{G}_i}$, is separable by subsystems, one can find a solution for~\eqref{eq:admm} parallelly by iterating over $\mathcal X_i$ and $\mathcal Z_i$, alternately, while leveraging new dual variables $\Lambda_i$. In summary, the ADMM algorithm~\cite{boyd2011distributed} updates the essential
variables at each iteration $k$ as follows:

\begin{itemize}
	\item For each $i \in \{1,\ldots,M\}$, solve the following local problem:
	\begin{align*}
		\mathcal X_i^{k+1}\in \underset{\mathcal X^* \in \mathcal{G}_i}{\operatorname{argmin}}~~\|\mathcal X_i^{*} - \mathcal Z_i^k + \Lambda_i^k\|^2_F.
	\end{align*}
	\item If $(\mathcal X_1^{k+1},\ldots, \mathcal X_M^{k+1}) \in \mathcal{G}$, the algorithm is successfully terminated. Otherwise, solve the following global problem:
	\begin{align*}
		(\mathcal Z_1^{k+1},\!\ldots,\! \mathcal Z_M^{k+1})\!\in\! \!\!\!\!\underset{(\mathcal Z_1^*,\ldots, \mathcal Z_M^*) \in \mathcal{G}}{\operatorname{argmin}}\sum_{i = 1}^M \|\mathcal X_i^{k+1} \!-\! \mathcal Z_i^* \!+\! \Lambda_i^k\|^2_F.  
	\end{align*}
	
	\item Update dual variables $\Lambda_i$ as
	\begin{align*}
		&\Lambda_{i}^{k+1}=\mathcal X_i^{k+1}- \mathcal Z_i^{k+1} + \Lambda_{i}^k,
	\end{align*}
	and return to the first step until a possible convergence.
\end{itemize}

\begin{remark}
	The data-driven results proposed in this work have the potential to be extended for designing a \emph{control Lyapunov function (CLF)} for unknown homogeneous interconnected systems. In such a case, the optimization programs are no longer convex due to the bilinearity between the unknown coefficients of the CLF and the controller. One possible solution is to treat the control inputs as discrete values (similar to switched systems), transforming the non-convex optimization problem into a mixed-integer linear programming.
\end{remark}

\begin{remark}
	Note that in the case of an unknown interconnection topology, while considering the entire state measurement of subsystems within $w_i$, as described in Remark~\ref{Rem:10}, does not necessitate full interconnection, it potentially increases the computational complexity to some extent. Specifically, this assumption raises the sample complexity in $\eqref{New6}$ and accordingly affects the computation of the maximum distance between points on the unit sphere and the set of data points, as described in $\eqref{EQ:41}$. Additionally, this assumption impacts the estimation of the Lipschitz constant in Algorithm~\ref{Alg:1}, as $\hat{w}_i^z$ also appears in this computation.
	In contrast, the results in Section~\ref{ADMM} do not involve these complexities but require knowledge of the interconnection topology. This highlights a trade-off: if the topology is known, it is advisable to use the results in Section~\ref{ADMM} (cf. the last case study). Otherwise, the results for an unknown topology offer a valuable solution to the problem, albeit with increased computational complexity (cf. the second case study).
\end{remark}

\section{Case Study}\label{Case_Study}
In this section, we begin by applying our results to a network composed of only two subsystems. We illustrate, utilizing both model-based techniques and our data-driven methods, that if both subsystems are dissipative in accordance with Definition~\ref{Def:4}, the interconnected network exhibits GAS. Subsequently, to demonstrate the effect of our data-driven method over networks with \emph{unknown} interconnection topology, we extend our analysis to a room temperature network with an unknown topology. Finally, to demonstrate that our data-driven results can be applied to \emph{nonlinear} dynamics with homogeneity, we present a nonlinear homogeneous network of degree one, consisting of $10,000$ subsystems (totaling $20,000$ dimensions), to highlight the applicability of our approach.

\begin{figure} 
	\begin{center}
		\includegraphics[width=0.7\linewidth]{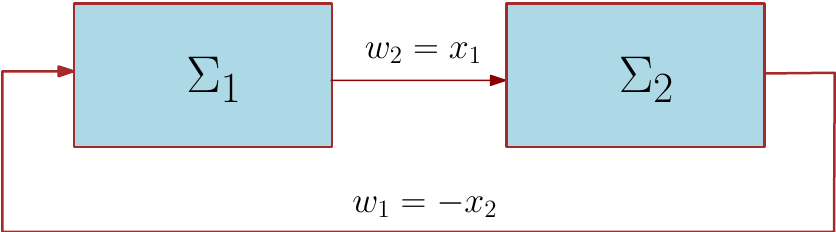} 
		\caption{Interconnected dt-NS $\mathcal{I}(\Sigma_1,\Sigma_2)$ with a coupling matrix $\mathcal M$ as in~\eqref{Coupling}.}
		\label{Fig3}
	\end{center}
\end{figure}

\subsection{Illustrative Example}
Consider two subsystems with the following dynamics:
\begin{align}\label{Sub}
	\Sigma_i\!: x_i(k+1)=A_i{x_i}(k)+B_iw_i(k), \quad i\in\{1,2\},
\end{align}
with system matrices
\begin{align*}
	A_1 =\begin{bmatrix} -0.1 & -0.2 \\ -0.2 &  -0.1 \end{bmatrix}\!, \quad B_1 =\begin{bmatrix} 0.2 & 0 \\ 0 &  0.2 \end{bmatrix}\!, \\
	A_2 =\begin{bmatrix} -0.5 & 1.8 \\ -0.2 &  1.01 \end{bmatrix}\!, \quad B_2 =\begin{bmatrix} 0.1 & 0 \\ 0 &  0.1 \end{bmatrix}\!.
\end{align*}
These two subsystems are interconnected through negative feedback, as illustrated in Fig.~\ref{Fig3}. It is evident that one can obtain the interconnected system $\Sigma := \mathcal{I}(\Sigma_1,\Sigma_2)$ as
\begin{align}\notag
	\Sigma\!:x(k+1)=A{x}(k),
\end{align}
with the block matrices
\begin{align*}
	A =\begin{bmatrix} A_1 & B_1 \\ B_2&  A_2 \end{bmatrix}\!\!,
\end{align*}
and the following interconnection matrix
\begin{align}\label{Coupling}
	\mathcal M =\begin{bmatrix} \begin{bmatrix} 0 & 0 \\ 0&  0 \end{bmatrix} & \begin{bmatrix} -1 & 0 \\ 0&  -1 \end{bmatrix} \\ \begin{bmatrix} 1 & 0 \\ 0&  1 \end{bmatrix} &  \begin{bmatrix} 0 & 0 \\ 0&  0 \end{bmatrix} \end{bmatrix}\!\!.
\end{align}

We initially presume access to the model dynamics and demonstrate that each subsystem $\Sigma_i$ displays dissipative behavior, fulfilling condition~\eqref{Eq:8_2}, through the utilization of a positive-definite storage function $\mathcal S_i.$ We then show that under the model-based compositionality condition in~\eqref{EQQ:2}, the interconnected system $\Sigma := \mathcal{I}(\Sigma_1,\Sigma_2)$ is GAS in the sense of Definition~\ref{GAS} and according to Theorem~\ref{Thm:3}.

We first fix a candidate positive-definite storage function $\mathcal S_i$ as
\begin{align}\label{Storage}
	\mathcal S_i(q_i,x_i) = x_i^\top P_ix_i,
\end{align}
with $P$ being a $(2\times2)$ symmetric positive-definite matrix. We now aim at showing condition~\eqref{Eq:8_2}. Given that we are aware of the model at this stage, we can derive the following model-based analysis:
\begin{align*}
	\mathcal S_i&(q_i,f_i(x_i,w_i)) = (A_ix_i + B_i w_i)^\top P_i (A_ix_i + B_i w_i) \\
	& = x_i^\top A_i^\top P_iA_ix_i + w_i^\top B_i^\top P_iB_iw_i + 2x_i^\top A_i^\top P_iB_iw_i \\
	&=\begin{bmatrix}x_i\\w_i\\\end{bmatrix}^\top\begin{bmatrix}\notag
		A_i^\top P_iA_i & A_i^\top P_iB_i\\
		B_i^\top P_iA_i& B_i^\top P_iB_i\\
	\end{bmatrix}\begin{bmatrix}x_i\\w_i\\\end{bmatrix}\!\!.
\end{align*}
By assuming the following LMI, 
\begin{align}\label{LMI}
	\begin{bmatrix}
		A_i^\top P_iA_i & A_i^\top P_iB_i\\
		B_i^\top P_iA_i& B_i^\top P_iB_i\\
	\end{bmatrix}\leq \begin{bmatrix}
		P_i + \mathcal X_i^{22} & \mathcal X_i^{21}\\
		\mathcal X_i^{12}& \mathcal X_i^{11}
	\end{bmatrix}\!\!,
\end{align}
we have
\begin{align*}
	\mathcal S_i(q_i,f_i(x_i,w_i)) &= \begin{bmatrix}x_i\\w_i\\\end{bmatrix}^\top\begin{bmatrix}\notag
		A_i^\top P_iA_i & A_i^\top P_iB_i\\
		B_i^\top P_iA_i& B_i^\top P_iB_i\\
	\end{bmatrix}\begin{bmatrix}x_i\\w_i\\\end{bmatrix}\\
	& \leq \begin{bmatrix}x_i\\w_i\\\end{bmatrix}^\top\begin{bmatrix}\notag
		P_i +  \mathcal X_i^{22}  &  \mathcal X_i^{21}\\
		\mathcal X_i^{12}& \mathcal X_i^{11}
	\end{bmatrix}\begin{bmatrix}x_i\\w_i\\\end{bmatrix}\\
	& \leq \underbrace{x_i^\top P_ix_i}_{\mathcal S_i(q_i,x_i)} + \begin{bmatrix}w_i\\x_i\\\end{bmatrix}^\top\begin{bmatrix}\notag
		\mathcal X_i^{11}  & \mathcal X_i^{12}\\
		\mathcal X_i^{21}& \mathcal X_i^{22}
	\end{bmatrix}\begin{bmatrix}w_i\\x_i\\\end{bmatrix}\!,
\end{align*}
which is exactly the satisfaction of condition~\eqref{Eq:8_2}. Hence, to demonstrate the satisfaction of condition~\eqref{Eq:8_2} by the candidate storage function $\mathcal{S}_i$ as in~\eqref{Storage}, it is sufficient to solve the constructed LMI in~\eqref{LMI} for each subsystem. This involves the design of positive-definite matrices $P_i$ and storage functions $\mathcal{X}_i$ for both subsystems.

We use the semidefinite programming (SDP) solver \textsf{SeDuMi}~\cite{sturm1999using} and fulfill the LMI in~\eqref{LMI} for both subsystems with the following matrices:
\begin{align*}
	P_1 &=\begin{bmatrix} 1.7779 & -0.0290 \\ -0.0290 & 1.9546 \end{bmatrix}\!\!,  \mathcal X^{11}_1 =\begin{bmatrix} 0.3372 & -0.1620 \\ -0.1620 & 0.5128 \end{bmatrix}\!\!, \\
	\mathcal X^{12}_1 &\!=\!\mathcal X^{21}_1\!=\!\begin{bmatrix} -0.0250 \!\!&\!\! -0.0486 \\ -0.0486 \!\!&\!\!  0.0923 \end{bmatrix}\!\!, \mathcal X^{22}_1\!=\!\begin{bmatrix} -1.0748 \!\!&\!\!  0.0579 \\ 0.0579 \!\!&\!\!  -1.1826 \end{bmatrix}\!\!,\\
	P_2 &=\begin{bmatrix} 1.1742 & -1.4717 \\ -1.4717 &  5.8613 \end{bmatrix}\!\!,  \mathcal X^{11}_2=\begin{bmatrix} 0.5509 & -0.0283 \\ -0.0283 & 0.6199 \end{bmatrix}\!\!, \\		\mathcal X^{12}_2 &\!=\!\mathcal X^{21}_2\!=\!\begin{bmatrix} -0.0147 \!\!&\!\!  -0.0256 \\ -0.0256 \!\!&\!\!  0.1947 \end{bmatrix}\!\!, \mathcal X^{22}_2\!=\!\begin{bmatrix} -0.6036 \!\!&\!\!  0.3231 \\ 0.3231 \!\!&\!\!  -0.9405 \end{bmatrix}\!\!,
\end{align*}
certifying that each subsystem is dissipative in the sense of Definition~\ref{Def:4}. Now, we look at $\Sigma=\mathcal{I}(\Sigma_1,\Sigma_2)$ with the coupling matrix $\mathcal M$ as in~\eqref{Coupling} aiming to satisfy compositionality condition \eqref{EQQ:2}. Given the obtained supply rates $\mathcal X_i$, one can construct matrix $\mathcal X_{cmp}$ in \eqref{EQQ:2} as
\begin{align*}
	\mathcal X_{cmp}&=\begin{bmatrix}\notag
		\mathcal X_1^{11}  & \mathbf{0}_{2\times 2} & \mathcal X_1^{12} & \mathbf{0}_{2\times 2}\\
		\mathbf{0}_{2\times 2} & \mathcal X_2^{11} & \mathbf{0}_{2\times 2} & \mathcal X_2^{12}\\
		\mathcal X_1^{21}& \mathbf{0}_{2\times 2} & \mathcal X_1^{22} & \mathbf{0}_{2\times 2}\\
		\mathbf{0}_{2\times 2} & \mathcal X_2^{21}& \mathbf{0}_{2\times 2} & \mathcal X_2^{22}
	\end{bmatrix}\!\!.
\end{align*}
The compositionality condition \eqref{EQQ:2} is then satisfied as 
\begin{align}\notag
	\begin{bmatrix}
		\mathcal M\\\mathds{I}_4
	\end{bmatrix}^\top &\!\!\mathcal X_{cmp}\begin{bmatrix}
		\mathcal M\\\mathds{I}_4
	\end{bmatrix} =\\
	&\begin{bmatrix}   -0.5238 & 0.0297 & 0.0103 & 0.0230 \\ 0.0297 & -0.5627 & 0.0230 & 0.1024 \\ 0.0103 & 0.0230 & -0.2664 & 0.1611 \\ 0.0230 & 0.1024 & 0.1611 & -0.4277\end{bmatrix} \preceq 0,
\end{align}
concluding that the interconnected system $\Sigma=\mathcal{I}(\Sigma_1,\Sigma_2)$ is GAS in the sense of Definition~\ref{GAS} and according to Theorem~\ref{Thm:3}.

We now assume that we are not aware of the model of subsystems in~\eqref{Sub} and the interconnection topology in~\eqref{Coupling}, and employ our data-driven techniques to tackle the problem. We fix the structure of our storage functions as $\mathcal S_i(q_i,x_i) = q^1_{i}x_{1_i}^2 + q^2_{i}x_{2_i}^2$ for both  $i\in\{1,2\}$. We employ Algorithm~\ref{Alg:2} as our proposed data-driven scheme. We collect samples from trajectories of each unknown subsystem and normalize them to be projected onto unit sphere. We now solve the $\text{SCP}$~\eqref{SCP} with collected data and compute coefficients of storage functions together with other decision variables as
\begin{align*}
	&\mathcal{S}_1(q_1,x_1)\!=\! 0.1x_{1_1}^2 + 0.2x_{2_1}^2,~ \mu_{\mathcal N_1}^*\!=\!-0.0607,~ \delta_{\mathcal N_1}^*\!=\! 0.039,\\
	&\mathcal X^{*11}_1 \!=\! 0.0001\mathds{I}_2, ~\mathcal X^{*12}_1 \!\!=\! \mathcal X^{*21}_1 \!\!=\! \mathbf{0}_{2\times 2},~  ~\mathcal X^{*22}_1 \!\!=\!\begin{bmatrix} -0.02 \!\!&\!\! 0 \\ 0 \!\!&\!\! -0.01 \end{bmatrix}\!\!, \\
	&	\mathcal{S}_2(q_2,x_2)\!=\! 0.01x_{1_2}^2 + 0.4x_{2_2}^2, ~ \mu_{\mathcal N_2}^*\!=\! -0.0108,~ \delta_{\mathcal N_2}^*\!=\! 0.03,\\
	&\mathcal X^{*11}_2 \!=\! 0.0001\mathds{I}_2, ~\mathcal X^{*12}_2 \!\!=\! \mathcal X^{*21}_2 \!\!=\! \mathbf{0}_{2\times 2},~  ~\mathcal X^{*22}_2 \!\!=\!-0.001\mathds{I}_2.
\end{align*}
We compute $\varepsilon_1 = \varepsilon_2 = 0.0041, \mathscr{L}^1_{1} = 0.1714, \mathscr{L}^2_{1} = 0.1631, \mathscr{L}^1_{2} = 0.3906,$ and $\mathscr{L}^2_{2} = 0.1454$ according to Steps 2 and 3 in Algorithm~\ref{Alg:2}. Since 
\begin{align*}
	&\mu_{\mathcal N_1}^* \!+\! \mathscr{L}^1_{1} \varepsilon_{1} = -0.06 < 0,\\
	&\mu_{\mathcal N_2}^* \!+\! \mathscr{L}^1_{2} \varepsilon_{2} = -0.0091 < 0,\\
	&\sum_{i=1}^2\big(\mu_{\mathcal N_i}^* \!+\! \delta_{\mathcal N_i}^* \!+\! \mathscr{L}^2_{i} \varepsilon_{i} \big) = -0.0210 + 0.0198\\
	&\quad \quad \quad \quad\quad\quad\quad\quad\quad\quad= -0.0012 < 0,
\end{align*}
according to Theorem~\ref{Thm1}, one can certify that the interconnected dt-NS $\Sigma=\mathcal{I}(\Sigma_1,\Sigma_2)$ is GAS with respect to $x = 0$, and $$\mathcal  V(q, x) = \sum_{i=1}^{2}\mathcal S_i(q_i,x_i) = 0.1x_{1_1}^2 + 0.2x_{2_1}^2 + 0.01x_{1_2}^2 + 0.4x_{2_2}^2$$ is a valid Lyapunov function for the interconnected dt-NS with correctness guarantees. This is compatible with the results of the model-based approach that we acquired in the first step of this illustrative example.

\subsection{Room Temperature Network with Unknown Topology}
To demonstrate the scalability of our method, we apply our data-driven findings on a room temperature network with an unknown topology. The evolution of the temperature $T(\cdot)$ can be described by the following interconnected network~\cite{meyer}:
\begin{align}\notag
	\Sigma\!:T(k+1)=A{T}(k),
\end{align}
where $A$ is a matrix with diagonal elements $ a_{ii}=1-2\varphi-\theta$, $i\in\{1,\ldots, M\}$, and off-diagonal elements $a_{ij}=\varphi$, $i,j\in \{1,\ldots, M\}, i\neq j$, depending on the unknown interconnection topology.
In addition, $ T(k)=[T_1(k);\ldots;T_{M}(k)]$, and $\varphi$, $\theta$ are thermal factors between rooms $i$ and $j$, and the external environment and the room $i$, respectively. Now by characterizing each individual room as 
\begin{align*}
	\Sigma_i\!:T_i(k+1)=a_{ii}{T_i}(k)+\varphi w_{i}(k),
\end{align*}
with $w_i$ being defined according to~\eqref{Eq:4} with an unknown interconnection topology $\mathcal M$, one has $\Sigma=\mathcal{I}(\Sigma_1,\ldots,\Sigma_{M})$. We assume that the underlying model is unknown to us.

Our primary objective is to systematically construct a Lyapunov function for the unknown interconnected dt-NS by leveraging data-driven storage functions of individual subsystems through solving $\text{SCP}$~\eqref{SCP}. Subsequently, we confirm that the interconnected network achieves GAS with respect to its equilibrium point $x = 0$, ensuring a correctness guarantee.

Initially, we establish the structure of our storage functions as $\mathcal S_i(q_i,x_i) = q_{i}x_i^2$ for all $i\in\{1,\ldots,20\}$. We employ our proposed data-driven scheme, outlined in Algorithm~\ref{Alg:2}. We gather samples from trajectories of each unknown room and normalize them for projection onto a unit sphere. Subsequently, we solve the $\text{SCP}$~\eqref{SCP}, computing coefficients of storage functions alongside other decision variables as follows:
\begin{align*}\notag
	&\mathcal{S}^*_i(q_i,x_i)= 1.7x_i^2, ~\mu_{\mathcal N_i}^*=-0.1013, ~\delta_{\mathcal N_i}^*=0.02,\\
	&\mathcal X^{*11}_i =\begin{bmatrix}0.0001 & 0 & \cdots & \cdots & 0 \\  0 & 0.0001 & 0 & \cdots & 0 \\ 0 & 0 & 0 & \cdots & 0\\ \vdots &  & \ddots & \ddots & \vdots \\ 0& \cdots & \cdots & 0 & 0\end{bmatrix}_{19\times 19}\!\!\!\!\!\!\!\!\!\!\!\!\!\!\!\!, \\
	&\mathcal X^{*12}_i = \mathbf{0}_{19\times 1},~  \mathcal X^{*21}_i \!\!=\! \mathcal X^{*12^\top}_i\!\!, ~\mathcal X^{*22}_i = -0.001.
\end{align*}
We compute $\varepsilon_i = 0.0091$, $\mathscr{L}^1_{i} = 1.3317,$ and $\mathscr{L}^2_{i} = 0.3617$  according to Steps 3, 4 in Algorithm~\ref{Alg:2}. Since 
\begin{align*}
	&\mu_{\mathcal N_i}^* \!+\! \mathscr{L}^1_{i} \varepsilon_{i} = -891 \times 10^{-4} < 0,\quad \forall i\in\{1,\dots,20\}, \\
	&\sum_{i=1}^{20}(\mu_{\mathcal N_i}^* \!+\! \delta_{\mathcal N_i}^* \!+\! \mathscr{L}^2_{i} \varepsilon_{i}) = \sum_{i=1}^{20}(-780 \times 10^{-4})_i  = -1.56 < 0,
\end{align*}
according to Theorem~\ref{Thm1}, one can verify that unknown interconnected dt-NS $\Sigma=\mathcal{I}(\Sigma_1,\dots,\Sigma_{20})$ is GAS with respect to $x = 0$, and $$\mathcal  V(q, x) = \sum_{i=1}^{20}\mathcal S_i(q_i,x_i) = \sum_{i=1}^{20} 1.7x_i^2$$ is a valid Lyapunov function for the interconnected dt-NS with a correctness guarantee.

\begin{remark}
	Note that certain physical insights about unknown models can help guide the selection of homogeneous basis functions. For instance, since room temperature models are typically polynomial due to their underlying physics, we chose the basis functions to be monomials based on the state. Regarding the number of monomial elements, one can include all combinations of state variables (while preserving homogeneity), and the scenario optimization problem will automatically eliminate some of them during the solution of the SCP by assigning zero to their coefficients.
\end{remark}

\begin{remark}
	The term ``unknown topology'' in our work does not imply that our data-driven techniques can handle any arbitrary interconnection topology. In fact, for model-based dissipativity reasoning, knowledge of the topology is required to verify the compositionality condition in \eqref{EQQ:2}. However, since our data-driven conditions in~\eqref{Con1}-\eqref{Con2} are different from model-based techniques, the interconnection topology does not require to be known. While the unknown topology only influences $w_i$ in the SCP during the solution process, this does not imply that if the SCP is solvable, it will be so for any arbitrary interconnection topology. For each new unknown topology, a new set of data would need to be collected, followed by an attempt to solve the SCP.
\end{remark}

\begin{figure} 
	\begin{center}
		\includegraphics[width=1.05\linewidth]{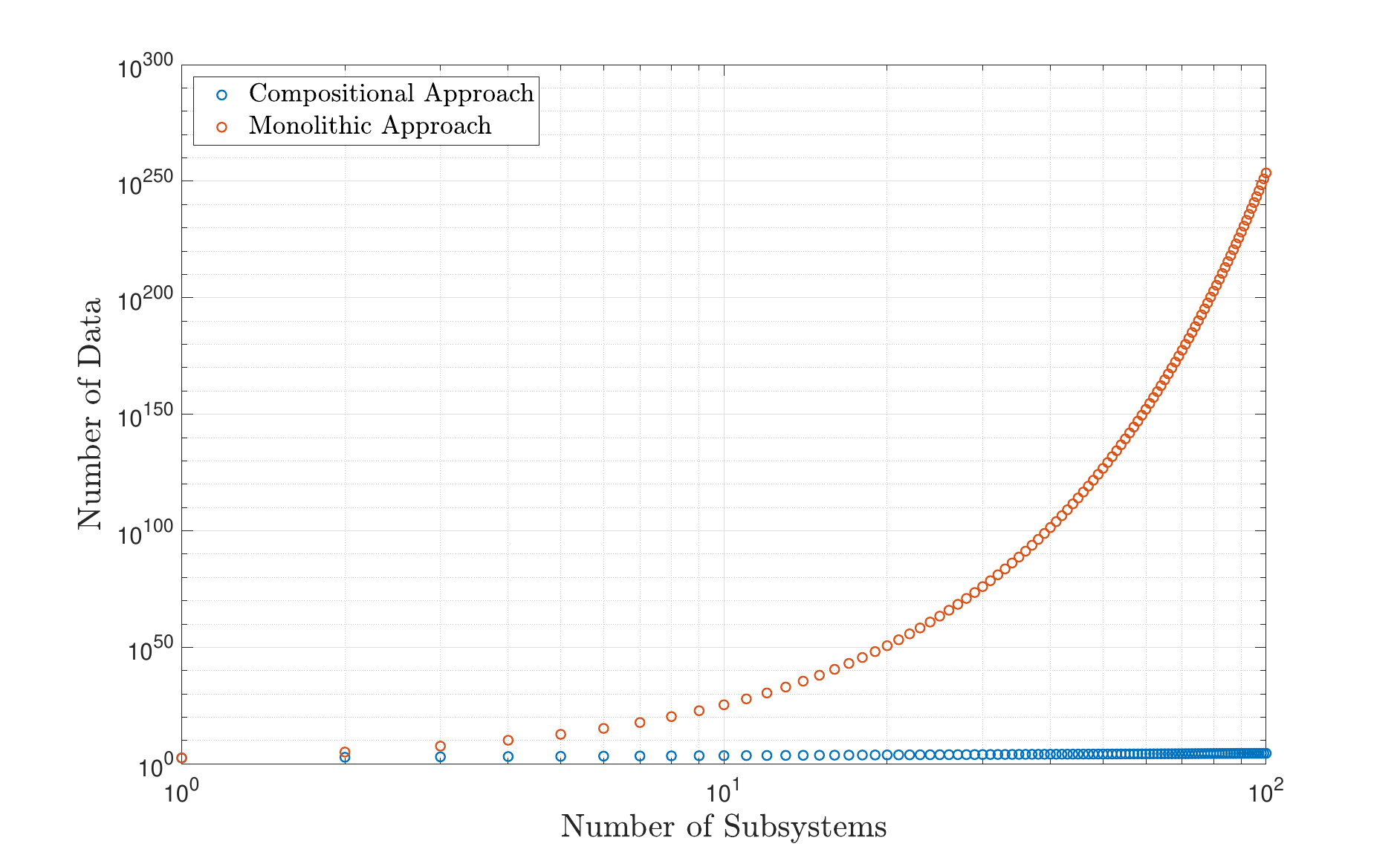} 
		\caption{Sample complexity analysis: Required amount of data vs. number of subsystems in monolithic and compositional approaches. Plot is in logarithmic scale. In our data-driven compositional approach, the sample complexity is reduced to the level of subsystems, leading to a \emph{linear} growth in the number of required data as the number of subsystems increases. However, in the monolithic approach, the sample complexity grows \emph{exponentially} with the number of subsystems, rendering it impractical in practice.} 
		\label{Fig2}
	\end{center}
\end{figure}

\subsection{Sample Complexity Analysis}
{In this subsection, we delve into a comprehensive analysis of sample complexity within the frameworks of both monolithic and compositional approaches. To do so, we have visually represented the relationship between the required amount of data and the number of subsystems for both methodologies in Fig.~\ref{Fig2}. It becomes evident that our data-driven compositional approach significantly mitigates sample complexity by aligning it with the granularity of subsystems. Consequently, as the number of subsystems increases, the growth in required data remains \emph{linear}. Conversely, in the monolithic approach (\emph{i.e.,} solving the problem directly over the network), sample complexity escalates \emph{exponentially} in tandem with the number of subsystems, making it operationally unfeasible.}

\subsection{Nonlinear Homogeneous Network}\label{third}
{To demonstrate that our data-driven results can be applied to nonlinear dynamics with homogeneity, we present here a \emph{nonlinear homogeneous network} of degree one, consisting of $10,000$ subsystems (totaling $20,000$ dimensions), to highlight the applicability of our approach for nonlinear systems with high dimensions. We assume that the interconnection topology is known as a unidirectional ring to use the results of Section~\ref{ADMM}.}

{The evolution of each subsystem can be described by the following nonlinear homogeneous system:
	\begin{align*}
		\Sigma_i\!:x_i(k\!+\!1)\!=\!H_i\mathcal F_i(x_i(k))+0.1w_{i}(k),
	\end{align*}
	where $x_i = [x_{1_i}; x_{2_i}]$ and
	\begin{align}\notag
		H_i &= \begin{bmatrix} -0.1 & 0.2 \\ 0.15 & 0.12 \end{bmatrix}\!\!, \\\label{new9}
		\mathcal F_i(x_i(k)) &= \begin{bmatrix} \sqrt{|x_{1_i}(k)x_{2_i}(k)| + \gamma(x_{1_i}^2(k) + x_{2_i}^2(k))} \\ x_{2_i} \end{bmatrix}\!\!,
	\end{align}
	where $\gamma = 0.1$, and $w_i$ is defined according to~\eqref{Eq:4} for a unidirectional ring interconnection topology, \emph{i.e.,} $w_i = x_{i-1}$ with $x_0 = x_M$.}

{It is clear that each subsystem $\Sigma_i$ is homogeneous of degree one, \emph{i.e.,} for any $\lambda > 0$ and $x_i\in \mathbb R^{n_i},w_i\in \mathbb R^{p_i}$, $f_i(\lambda x_i, \lambda w_i) = \lambda f_i(x_i,w_i)$:
	\begin{align*}
		f_i(\lambda x_i,& \lambda w_i) \\
		&=H_i\begin{bmatrix} \sqrt{|\lambda x_{1_i}(k)\lambda x_{2_i}(k)|  \!+ \! \gamma(\lambda^2x_{1_i}^2(k) \!+ \! \lambda^2x_{2_i}^2(k))} \\ \lambda x_{2_i} \end{bmatrix} \\
		&~~~+ 0.1(\begin{bmatrix} \lambda w_{1_{i}}(k)\\ \lambda w_{2_{i}}(k) \end{bmatrix})\\
		&=  \lambda  H_i\begin{bmatrix} \sqrt{|x_{1_i}(k)x_{2_i}(k)| + \gamma(x_{1_i}^2(k) + x_{2_i}^2(k))} \\  x_{2_i} \end{bmatrix} \\ &~~~+0.1\lambda(\begin{bmatrix}  w_{1_{i}}(k)\\  w_{2_{i}}(k) \end{bmatrix})\\
		&=  \lambda f_i(x_i,w_i).
\end{align*}}
{Now the dynamics of the interconnected network  $\Sigma=\mathcal{I}(\Sigma_1,\ldots,\Sigma_{M})$ with $M = 10000$ can be established as
	\begin{align}\notag
		\Sigma\!:x(k+1)=H\mathcal F(x(k)) + 0.1 \mathcal M x(k),
	\end{align}
	where $H$ is a matrix with block diagonal elements $H_i$ as in~\eqref{new9} and block off-diagonal zero matrices, while $\mathcal M$ is a unidirectional ring interconnection, with elements $m_{i+1,i} = \mathds{I}_2, i\in\{1,\dots,M-1\}$, $m_{1,M} = \mathds{I}_2$, and all other elements are $\mathbf{0}_{2\times 2}$. In addition, $ x(k)=[x_1(k);\ldots;x_{M}(k)]$, and $\mathcal F(x(k))=[\mathcal F_1(x_1(k));\ldots;\mathcal F_M(x_M(k))]$. Given that each subnetwork is homogeneous of degree one, it is clear that the interconnected network also remains homogeneous of degree one, \emph{i.e.,} for any $\lambda > 0$ and $x\in \mathbb R^{n}$, $f(\lambda x) = \lambda f(x)$.}

{Initially, we establish the structure of our storage functions as $\mathcal S_i(q_i,x_i) = q_{1_i}x_{1_i}^2 + q_{2_i}x_{1_i}x_{2_i} + q_{3_i}x_{2_i}^2$ for all $i\in\{1,\ldots,10000\}$. It is clear that $\mathcal S_i(q_i,x_i)$ is a homogeneous function of degree $2$, \emph{i.e.,} for any $\lambda > 0$ and $x_i\in \mathbb R^{n_i}$, $\mathcal S_i(\lambda x_i) = \lambda^2 \mathcal S_i(x_i)$. We gather $1296$ samples from trajectories of each unknown subsystem (denoted as $\mathcal N_i = 1296$) and normalize them for projection onto a unit sphere. Subsequently, we solve the $\text{SCP}$~\eqref{SCP-1} with $\mathcal N_i = 1296$, computing coefficients of storage functions alongside other decision variables as follows, for any $i\in\{1,\dots,10000\}$:
	\begin{align*}
		&\mathcal{S}^*_i(q_i,x_i)= 1.6x_{1_i}^2 +1.2x_{1_i}x_{2_i} + 1.9x_{2_i}^2, \mu_{\mathcal N_i}^*=-0.777, \\
		&\mathcal X^{*11}_i \!\!=\! 0.0001\mathds{I}_2, \mathcal X^{*12}_i \!\!=\! \mathbf{0}_{2\times 2},  \mathcal X^{*21}_i \!\!=\! \mathcal X^{*12^\top}_i\!\!\!\!,\mathcal X^{*22}_i \!\!=\! -0.001\mathds{I}_2.
	\end{align*}
	We compute $\varepsilon_i = 0.0155$, $\mathscr{L}^1_{i} = 2.3436,$ and $\mathscr{L}^2_{i} =  2.3533$. Since 
	\begin{align*}
		&\mu_{\mathcal N_i}^* \!+\! \mathscr{L}_{i} \varepsilon_{i} = -7408\times 10^{-4} < 0, \quad \forall i\in\{1,\dots,10000\},\\
	\end{align*}
	according to Theorem~\ref{Thm2}, the constructed $\mathcal{S}_i$ obtained by solving SCP~\eqref{SCP-1} serve as storage functions for unknown subsystems $\Sigma_i$ globally over $\mathbb{R}^{n_i}$, with  guarantees of correctness. We also employ the ADMM approach proposed in Section~\ref{ADMM} and ensure that the dissipativity compositionality condition in $\eqref{EQQ:2}$ is satisfied. Then one can verify that unknown interconnected dt-NS $\Sigma=\mathcal{I}(\Sigma_1,\dots,\Sigma_{10000})$ is GAS with respect to $x = 0$, and $$\mathcal  V(q, x) = \sum_{i=1}^{10000}\mathcal S_i(q_{i},x_i) = \sum_{i=1}^{10000} 1.6x_{1_i}^2 +1.2x_{1_i}x_{2_i} + 1.9x_{2_i}^2$$ is a valid Lyapunov function for the interconnected dt-NS with a correctness guarantee. The computation of storage functions took $2.14$ seconds for each subsystem ($21400$ seconds for all $10000$ subsystems in a serial computation) on a MacBook Pro (Apple M2Max with 32 GB RAM Memory).}

{\section{Conclusion}\label{Discussion}
	In this work, we proposed a data-driven \emph{divide and conquer} approach for verifying the GAS property of interconnected nonlinear networks with a \emph{fully unknown topology}, while ensuring guaranteed confidence. The proposed framework utilized dissipativity-type properties of subsystems, which are characterized by the notion of storage functions. In our data-driven setting, we collected data from trajectories of each unknown subsystem and formulated a scenario convex program (SCP) to enforce the necessary conditions of storage functions. By solving the SCP using the collected data, we built a Lyapunov function for the interconnected system based on storage functions of individual subsystems, derived from data, by offering a newly developed data-driven compositionality condition in the dissipativity setting. To demonstrate the effectiveness of our data-driven approaches, we applied them to an unknown room temperature network consisting of $10000$ rooms with an unknown interconnection topology. As part of our future work, we aim to develop a data-driven approach for constructing \emph{control Lyapunov functions} for unknown nonhomogeneous interconnected systems with unknown topology. In addition, generalizing our data-driven approach to a broader class of nonlinear systems is being investigated as future work.}

\section*{REFERENCES}\vspace{-0.7cm}

\bibliographystyle{IEEEtran}
\bibliography{biblio}

\begin{IEEEbiography}[{\includegraphics[width=1in,height=1.25in,clip,keepaspectratio]{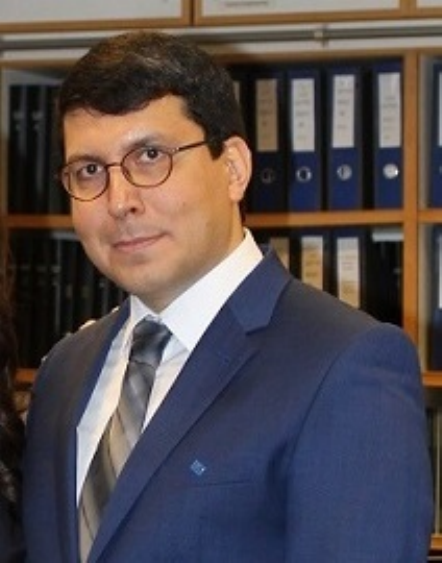}}]{Abolfazl Lavaei}~(M'17--SM'22)
		is an Assistant Professor in the School of Computing at Newcastle University, United Kingdom. Between January 2021 and July 2022, he was a Postdoctoral Associate in the Institute for Dynamic Systems and Control at ETH Zurich, Switzerland. He was also a Postdoctoral Researcher in the Department of Computer Science at LMU Munich, Germany, between November 2019 and January 2021. He received the Ph.D. degree in Electrical Engineering from the Technical University of Munich (TUM), Germany, in 2019. He obtained the M.Sc. degree in Aerospace Engineering with specialization in Flight Dynamics and Control from the University of Tehran (UT), Iran, in 2014. He is the recipient of several international awards in the acknowledgment of his work including  Best Repeatability Prize (Finalist) at the ACM HSCC 2025, IFAC ADHS 2024, and IFAC ADHS 2021, HSCC Best Demo/Poster Awards 2022 and 2020, IFAC Young Author Award Finalist 2019, and Best Graduate Student Award 2014 at University of Tehran with the full GPA (20/20). His research interests revolve around the intersection of Control Theory, Formal Methods in Computer Science, and Data Science.
\end{IEEEbiography}

\begin{IEEEbiography}[{\includegraphics[width=1in,height=1.25in,clip,keepaspectratio]{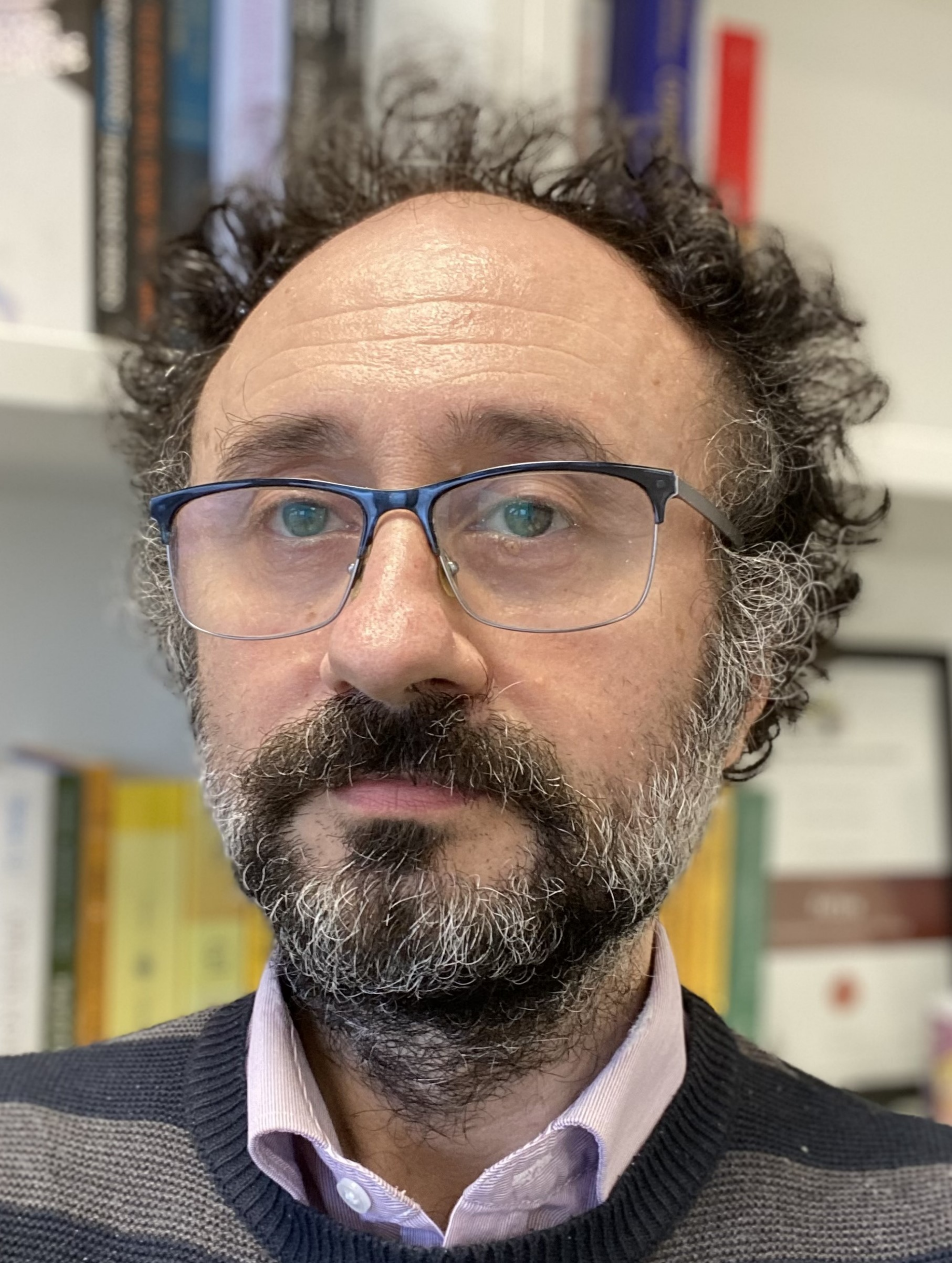}}]{David Angeli}~{(FIEEE) is a
		Professor of Nonlinear Networks Dynamics within the Dept of Electrical and Electronic Engineering of Imperial College London.  
		He received the B.S. degree in computer science engineering and the PhD degree in control theory from the University of Florence, Florence, Italy, in 1996 and 2000, respectively.
		Since 2000, he has been an Assistant Professor and since 2005, an Associate Professor with the Department of Information Engineering, University of Florence. In 2007, he was a Visiting Professor with I.N.R.I.A de Rocquencourt, Paris, France, and in 2008, he joined as a Senior Lecturer the Department of Electrical and Electronic Engineering, Imperial College London, London, U.K., where he is currently a Professor and the Director of Postgraduate Teaching. He is the author of more than 120 journal papers in the research areas of stability of nonlinear systems, control of constrained systems (MPC), chemical reaction networks theory, and smart grids.
		Prof Angeli was an Associate Editor for the IEEE Transactions in Automatic Control and Automatica and was elevated to Fellow of the IEEE in 2015 for contributions to nonlinear control theory. He is a Fellow of the IET since 2018 and the recipient of the Honeywell Medal from InstMC in 2021.}
\end{IEEEbiography}

\end{document}